\newif\ifshaphered

\ifdefined\isshaphered
\shapheredtrue
\fi

\documentclass[conference]{IEEEtran}

\IEEEoverridecommandlockouts
\usepackage{datetime}
\usepackage{authblk}

\ifCLASSOPTIONcompsoc
  \usepackage[nocompress]{cite}
\else
  \usepackage{cite}
\fi

\ifshaphered
\usepackage{longtable}
\newcommand\revision[1]{\textcolor{red}{#1}}

\newcommand\deleted[1]{\st{#1}}

\else

\newcommand\revision[1]{#1}

\newcommand\deleted[1]{}

\fi

\newcommand*{\textlabel}[2]{%
  \edef\@currentlabel{#1}
  \phantomsection
  #1\label{#2}
}
\usepackage{hyperref}

\usepackage[english]{babel}
\usepackage{blindtext}

\usepackage[utf8]{inputenc}
\usepackage{amsfonts} 
\usepackage{fontawesome}
\usepackage{amssymb}
\usepackage[ruled,lined,linesnumbered]{algorithm2e} \usepackage{soul}
\usepackage{algorithmic}
\usepackage{setspace}
\usepackage{graphicx}
\usepackage{multirow}
\usepackage{booktabs}

\usepackage{tikz}
\usetikzlibrary{positioning}

\usetikzlibrary{decorations.pathreplacing, automata}
\usetikzlibrary{arrows}
\usepackage{amsmath}
\usepackage{pifont}

\usetikzlibrary{shapes}
\makeatletter
\tikzset{circle split part fill/.style args={#1,#2}{%
 alias=tmp@name, 
  postaction={%
    insert path={
     \pgfextra{%
     \pgfpointdiff{\pgfpointanchor{\pgf@node@name}{center}}%
                  {\pgfpointanchor{\pgf@node@name}{east}}%
     \pgfmathsetmacro\insiderad{\pgf@x}
      \fill[#1] (\pgf@node@name.base) ([xshift=-\pgflinewidth]\pgf@node@name.east) arc
                          (0:180:\insiderad-\pgflinewidth)--cycle;
      \fill[#2] (\pgf@node@name.base) ([xshift=\pgflinewidth]\pgf@node@name.west)  arc
                           (180:360:\insiderad-\pgflinewidth)--cycle;            
         }}}}}
\usetikzlibrary{spy}
\usetikzlibrary{arrows}
\usepackage{amsmath}
\usepackage{pifont}

\usepackage{siunitx}
\usepackage{amsthm}
\usepackage{thmtools,thm-restate}
\newtheorem{corollary}{Corollary}
\usepackage{mdframed}

\usepackage{subfigure}

\usepackage{enumitem}
\setlist[enumerate]{noitemsep, topsep=2pt}
\setlist[itemize]{noitemsep, topsep=2pt}
\setlist[description]{noitemsep, topsep=2pt, font=\normalfont\space}

\usepackage{pgfplots}
\pgfplotsset{compat=1.15} 
\pgfplotsset{%
    axis line origin/.style args={#1,#2}{
        x filter/.append code={ 
            \ifx\pgfmathresult\empty\else\pgfmathparse{\pgfmathresult-#1}\fi
        },
        y filter/.append code={
            \ifx\pgfmathresult\empty\else\pgfmathparse{\pgfmathresult-#2}\fi
        },
        xticklabel=\pgfmathparse{\tick+#1}\pgfmathprintnumber{\pgfmathresult},
        yticklabel=\pgfmathparse{\tick+#2}\pgfmathprintnumber{\pgfmathresult}
    }
}
\usepackage{xspace}
\usepackage{numprint}

\npthousandsep{,}
\newif\ifshowcomment
\showcommentfalse

\newtheorem{definition}{Definition}

\newcommand{\Prov}{\mathcal{T}\xspace}
\newcommand{\Verif}{\mathcal{V}\xspace}
\newcommand{\Adv}{\mathcal{A}\xspace}
\newcommand{\Proof}{\mathcal{P}\xspace}

\newcommand{\WWW}{\mathbb{W}\xspace}
\newcommand{\III}{\mathbb{I}\xspace}
\newcommand{\HHH}{\mathbb{H}\xspace}
\newcommand{\AAA}{\mathbb{A}\xspace}

\newcommand{\XX}{\mathbf{X}\xspace}
\newcommand{\xx}{\mathbf{x}\xspace}
\newcommand{\yy}{\mathbf{y}\xspace}
\newcommand{\RR}{\mathbf{R}\xspace}

\newcommand{\fig}{\textrm{Figure}\xspace}

\newcommand{\Paragraph}[1]{~\vspace*{-0.8\baselineskip}\\{\bf #1}}

\usepackage{balance}

\IEEEoverridecommandlockouts

\begin{document}

\title{``Adversarial Examples'' for Proof-of-Learning}

\author{Rui Zhang$^\dag$\thanks{$^\dag$Rui Zhang and Jian Liu are co-first authors.}}
\author{Jian Liu$^\dag$\thanks{\IEEEauthorrefmark{1}Jian Liu is the corresponding author.}\IEEEauthorrefmark{1}}
\author{Yuan Ding}
\author{Zhibo Wang}
\author{Qingbiao Wu}
\author{Kui Ren}
\affil{Zhejiang University \authorcr Email: {\tt \{zhangrui98, liujian2411, dy1ant, zhibowang, qbwu, kuiren \}@zju.edu.cn}\vspace{1.5ex}}

\maketitle

\begin{abstract}
In S\&P '21, Jia et al. proposed a new concept/mechanism named proof-of-learning (PoL), which allows a prover to demonstrate ownership of a machine learning model by  proving integrity of the training procedure.
It guarantees that an adversary cannot construct a valid proof with less cost (in both computation and storage) than that made by the prover in generating the proof. 

A PoL proof includes a set of intermediate models recorded during training, together with the corresponding data points used to obtain each recorded model.
Jia et al. claimed that an adversary merely knowing the final model and training dataset cannot efficiently find a set of intermediate models with correct data points.

In this paper, however, we show that PoL is vulnerable to ``adversarial examples''! 
Specifically, in a similar way as optimizing an adversarial example, we could make an arbitrarily-chosen data point ``generate'' a given model, hence efficiently generating intermediate models with correct data points.
We demonstrate, both theoretically and empirically, that we are able to generate a valid proof with significantly less cost than generating a proof by the prover\deleted{, thereby we successfully break~PoL}\textlabel{.}{r4:1:1}

\end{abstract}

\ifshaphered
\input{shepherd-changelog}
\else
\fi

\section{Introduction}

Recently, Jia et al.~\cite{PoL} propose a concept/mechanism named {\em proof-of-learning} (PoL), 
which allows a prover $\Prov$ to prove that it has performed a specific set of computations to train a machine learning model;
and a verifier $\Verif$ can verify correctness of the proof with significantly less cost than training the model.
This mechanism can be  applied in at least two settings. 
First, when the intellectual property of a model owner is infringed upon (e.g., by a model stealing attack~\cite{tramer2016stealing, wang2018stealing, orekondy2019knockoff}),
it allows the owner to claim ownership of the model and resolve the dispute.
Second, in the setting of federated learning~\cite{mcmahan17a}, where a model owner distributes the training process across multiple workers, it allows the model owner to verify the integrity of the computation performed by these workers.
This could prevent Byzantine workers from conducting denial-of-service attacks~\cite{NIPS2017_f4b9ec30}.

\Paragraph{PoL mechanism.}
In their proposed mechanism~\cite{PoL}, 
$\Prov$ provides a PoL proof that includes:
(i) the training dataset,
(ii) the intermediate model weights at periodic intervals during training $W_0, W_k, W_{2k}, ..., W_{T}$, and 
(iii) the corresponding indices of the data points used to train each intermediate model.
With a PoL proof, one can replicate the path all the way from the initial model weights $W_0$ to the final model weights $W_T$ to be fully confident that $\Prov$ has indeed performed the computation required to obtain the final model.

During verification, $\Verif$ first verifies the provenance of the initial model weights $W_0$: whether it is sampled from the required initialization distribution;
and then recomputes a subset of the intermediate models to confirm the validity of the sequence provided.
However, $\Verif$ may not be able to reproduce the same sequence due to the noise arising from the hardware and low-level libraries.
To this end, they allow a distance between the recomputed model and its corresponding model in PoL.
Namely, for any $W_t$, $\Verif$ performs a series of $k$ updates to arrive at $W'_{t+k}$, which is compared to the purported $W_{t+k}$.
They tolerate:
\begin{center}
    $d(W_{t+k}, W'_{t+k}) \leq \delta$,
\end{center}
where $d$ represents a distance that could be $l_1$, $l_2$, $l_{\infty}$ or $cos$,
and $\delta$ is the verification threshold that should be calibrated before verification starts.



Jia et al.~\cite{PoL} claimed in their paper that an adversary $\Adv$ can never construct a valid a PoL with less cost (in both computation and storage) than that made by $\Prov$ in generating the proof (a.k.a. {\em spoof a PoL}).
However, they did not provide a proof to back their claim\textlabel{.}{r4:1:2}
\deleted{Instead, they simply designed some attacks by themselves and showed that those attacks are invalid. 
Without a doubt, this kind of security evaluation is unable to cover all potential attacks.}

\Paragraph{Our contribution.}
By leveraging the idea of generating adversarial examples, we successfully spoof a PoL!

In the PoL threat model, Jia et al.~\cite{PoL} assumed that ``{\em $\Adv$ has full access to the training dataset, and can modify it}''.
Thanks to this assumption, we can slightly modify a data point so that it can update a model and make the result pass the verification. 
In more detail, given the training dataset and the final model weights $W_T$, 
$\Adv$ randomly samples all intermediate model weights in a PoL: $W_0, W_k, W_{2k} ...$ (only $W_0$ needs to be sampled from the given distribution).
For any two neighboring model weights 
$(W_{t-k}, W_{t})$,
$\Adv$ picks batches of data points $(\mathbf{X}, \mathbf{y})$ from $D$,
and keeps manipulating $\mathbf{X}$ until:
\begin{center}
    $d(\texttt{update}(W_{t-k}, (\XX, \yy)), W_{t}) \leq \delta$.
\end{center}
The mechanism for generating adversarial examples ensures that the noise added to $\XX$ is minimized.

We further optimize our attack by sampling $W_0, W_k, W_{2k} ...$ in a way such that:
\begin{center}
    $d(W_t, W_{t-k}) < \delta$,  $\forall~0 < t < T$ and $t\mod k =0$.
\end{center}
With this condition, it becomes much easier for the ``adversarial'' $\XX$ to converge,
hence making our attack more efficient.

We empirically evaluate our attacks in both reproducibility and spoof cost.
We reproduced the results in~\cite{PoL} as baselines for our evaluations.
Our experimental results show that, in most cases of our setting, our attacks introduce smaller reproduction errors and less cost than the baselines\textlabel{.}{r4:1:3}
\deleted{That is to say, under the same assumption as the original paper, we can successfully spoof a PoL.}

\Paragraph{Organization.}
In the remainder of this paper, we first provide a brief introduction to PoL in Section~\ref{sec:background}.
Then, we formally describe our attack in Section~\ref{sec:attack} and extensively evaluate it in Section~\ref{sec:eval}.
In Section~\ref{sec:discuss}, 
we provide some countermeasures. 
Section~\ref{sec:related} compares our attacks to closely related work.

\Paragraph{Notations.} We introduce new notations as needed. 
A summary of frequent notations appears in Table~\ref{notationtable}.

\begin{table}[ht]
\small
\centering
\caption{Summary of notations}
\begin{spacing}{1.30}
\begin{tabular}{p{1.75cm} p{5cm}}
\hline
\textbf{Notation} & \textbf{Description} \\ \hline
$\Prov$ & prover \\ \hline
$\Verif$ & verifier \\ \hline
$\Adv$ & attacker \\ \hline
$D$ & dataset \\ \hline
$f_{W}$ & machine learning model \\ \hline
${W}$ & model weights \\ \hline
$E$ & number of epochs \\ \hline
$S$ & number of steps per epoch \\ \hline
$T$ & number of steps in $\Proof(\Prov, f_{W_T})$ \\ 
&$T = E \cdot S$ \\ \hline
$T'$ & number of steps in $\Proof(\Adv, f_{W_T})$ \\ \hline
$Q$ & number of models verified per epoch \\ \hline
$N$ & number of steps in generating an ``adversarial example'' \\ \hline
$k$ & number of batches in a checkpointing interval \\ \hline
$d()$ & distance that could be $l_1$, $l_2$, $l_{\infty}$ or $cos$ \\ \hline
$\delta$ & verification threshold \\ \hline
$\gamma$ & $\gamma \ll \delta$  \\ \hline
$\zeta$ & distribution for $W_0$ \\ \hline
$\eta$ & learning rate \\ \hline
$\varepsilon$ & reproduction error \\ \hline
$\mathbf{X}$ & batch of data points \\ \hline
$\mathbf{y}$ & batch of labels \\ \hline
$\mathbf{R}$ & batch of noise \\ \hline
\end{tabular}
\end{spacing}
\label{notationtable}
\vspace{-3mm}
\end{table}

\section{Proof-of-Learning}
\label{sec:background}

In this section, we provide a brief introduction to proof-of-learning (PoL). 
We refer to~\cite{PoL} for more details

\subsection{PoL definition}
\label{sec:definition}

PoL allows a prover $\Prov$ to demonstrate ownership of a machine learning model by  proving the integrity of the training procedure.
Namely, during training, $\Prov$ accumulates some secret information associated with training, which is used to construct the PoL proof $\Proof(\Prov, f_{W_T})$.
When the integrity of the computation (or model ownership) is under debate, an honest and trusted verifier $\Verif$ validates $\Proof(\Prov, f_{W_T})$ by querying $\Prov$ for a subset (or all of) the secret information, under which $\Verif$ should be able to ascertain if the PoL is valid or not.
A PoL proof is formally defined as follows:
\begin{definition}
A PoL proof generated by a prover $\Prov$ is defined as $\Proof(\Prov, f_{W_T}) = (\WWW, \III, \HHH, \AAA)$, 
where 
(a) $\WWW$ is a set of intermediate model weights recorded during training,
(b) $\III$ is a set of information about the specific data points used to train each intermediate model, 
(c) $\HHH$ is a set of signatures generated from these data points, and 
(d) $\AAA$ incorporates auxiliary information training the model such as hyperparameters, model architecture, optimizer and loss choices\footnote{For simplicity, we omit $\AAA$ in this paper and denote a PoL proof as $\Proof(\Prov, f_{W_T}) = (\WWW, \III, \HHH)$.}.
\end{definition}

An adversary $\Adv$ might wish to spoof $\Proof(\Prov, f_{W_T})$ by spending less computation and storage than that made by $\Prov$ in generating the proof.
By spoofing, $\Adv$ can claim that it has performed the computation required to train $f_{W_T}$. 
A PoL mechanism should guarantee: 
\begin{itemize}
    \item The  cost of verifying the PoL proof by $\Verif$ should be smaller than the cost (in both computation and storage) of generating the proof by $\Prov$.

    \item The cost of any spoofing strategy attempted by any $\Adv$ should be larger than the cost of generating the proof.
\end{itemize}

\subsection{Threat Model}
\label{sec:model}


In~\cite{PoL}, any of the following cases is considered to be a successful spoof by $\Adv$:
\begin{enumerate}
    \item {\em Retraining-based spoofing:} $\Adv$ produced a PoL for $f_{W_T}$ that is exactly the same as the one produced by $\Prov$, i.e., $\Proof(\Adv, f_{W_T}) = \Proof(\Prov, f_{W_T})$.
    \item {\em Stochastic spoofing:} $\Adv$ produced a valid PoL for $f_{W_T}$, but it is different from the one produced by $\Prov$ i.e.,$\Proof(\Adv, f_{W_T}) \neq \Proof(\Prov, f_{W_T})$.
    \item {\em Structurally Correct Spoofing:} $\Adv$ produced an invalid PoL for $f_{W_T}$ but it can pass the verification.
    \item {\em Distillation-based Spoofing:} $\Adv$ produced a valid PoL for an approximated model, which has the same run-time performance as $f_{W_T}$.
\end{enumerate}

The following adversarial capabilities are assumed in~\cite{PoL}:
\begin{enumerate}
    \item $\Adv$ has full knowledge of the model architecture, model weights, loss function and other hyperparameters.
    \item  $\Adv$ has full access to the training dataset $D$ and  can modify it. 
    {\bf This assumption is essential to our attacks.}
    \item 
    $\Adv$ has no knowledge of $\Prov$'s strategies about batching, parameter initialization, random generation and so on.
\end{enumerate}

\subsection{PoL Creation}

\renewcommand{\algorithmiccomment}[1]{$\triangleright$ #1}
\begin{algorithm}[htb]
\caption{PoL Creation (taken from~\cite{PoL})}
\label{alg:creation}
\LinesNumbered 
\KwIn {$D$, $k$, $E$, $S$, $\zeta$}
\KwOut {PoL proof: $\Proof(\Prov, f_{W_T})=(\WWW, \III, \HHH)$} 
$\WWW\leftarrow\{\}$
$\III\leftarrow\{\}$
$\HHH\leftarrow\{\}$\\
$W_0 \leftarrow \texttt{init}(\zeta))$ \hfill\CommentSty{initialize $W_0$}\\
\For{$e = 0  \to E-1 $}{
    $I \leftarrow \texttt{getBatches}(D, S)$
    
    \For{$s = 0  \to S-1 $}{
        $t:= e\cdot S + s$
        
        $W_{t+1} \leftarrow \texttt{update}(W_t, D[I[s]])$
        
        $\III.\texttt{append}(I[s])$
        
        $\HHH.\texttt{append}(h(D[I[s]]))$ \hfill\CommentSty{\revision{$h()$ is for computing the signature}}
        
        \eIf{$t \mod k = 0$}{
            $\WWW.\texttt{append}(W_t)$

        }{
            $\WWW.\texttt{append}(\mathbf{nil})$
        }
        
    }
}
\end{algorithm}




    
        
        
        
        

        

The PoL creation process is shown in Algorithm~\ref{alg:creation}, which is taken from~\cite{PoL} and slightly simplified by us.
$\Prov$ first initializes the weights $W_0$ according to an initialization strategy $\texttt{init}(\zeta)$ (line 2), where $\zeta$ is the distribution to draw the weights from.
If the initial model is obtained from elsewhere, a PoL is required for the initial model itself as well. 
We omit this detail in our paper for simplicity.

For each epoch, $\Prov$ gets $S$ batches of data points from the dataset $D$ via $\texttt{getBatches}(D, S)$ (Line 4), the output of which is a list of $S$ sets of data indices. 
In each step $s$ of the epoch $e$, the model weights are updated with a batch of data points in $D$ indexed by $I[s]$ (Line 7).
The \texttt{update} function leverages a suitable optimizer implementing a variant of gradient descent.
$\Prov$ records the updated model $W_t$ for every $k$ steps (Line 11), hence $k$ is a parameter called checkpointing interval and $\frac{1}{k}$ is then the checkpointing frequency.
To ensure that the PoL proof will be verified with the same data points as it was trained on, 
$\Prov$ includes a signature of the training data (Line 9) along with the data indices (Line 8).

\subsection{PoL Verification}
\label{sec:verification}

\renewcommand{\algorithmiccomment}[1]{$\triangleright$ #1}
\begin{algorithm}[htb]
\caption{PoL Verification (taken from~\cite{PoL})}
\label{alg:verification}
\KwIn {$\Proof(\Prov, f_{W_T})$, $D$, $k$, $E$, $S$, $\zeta$}
\KwOut {$\mathbf{success}$ / $\mathbf{fail}$} 
\If{\texttt{verifyInitialization}$(\WWW[0]) = \mathbf{fail}$}{
    {\bf return} $\mathbf{fail}$
}

$e \leftarrow 0$

$\mathit{mag} \leftarrow \{\}$

\For{$t = 0  \to T-1 $}{
    \If{$t \mod k = 0~\wedge~t \neq 0$}{
        $\mathit{mag}.\texttt{append}(d(\WWW[t], \WWW[t-k]))$
    }
    $e_t = \left \lfloor \frac{t}{S} \right \rfloor$
    
    \If{$e_t = e + 1$}{
        $idx \leftarrow \texttt{sortedIndices}(\mathit{mag}, \downarrow)$
        
        \If{\texttt{verifyEpoch}$(idx) = \mathbf{fail}$}{
            {\bf return} $\mathbf{fail}$
        }
        \revision{$e\leftarrow e_t$, $\mathit{mag} \leftarrow \{\}$}
    }
    \deleted{$e\leftarrow e_t$}
    
    \deleted{$\mathit{mag} \leftarrow \{\}$}
}
{\bf return} $\mathbf{success}$ \\
~
\SetKwFunction{FMain}{verifyEpoch}
\SetKwProg{Fn}{function}{}{end}

\Fn{\FMain{$idx$}}{
    \For{$q = 1 \to Q$}{
        $t := idx[q-1]$
        
        
        
        \deleted{$\texttt{verifyDataSignature}(\HHH[t], \III[t])$}
        
        \revision{$\texttt{verifyDataSignature}(\HHH[t], D[\III[t]])$}
        
        $W'_t \leftarrow \WWW[t]$
        
        \For{$i = 0 \to (k-1)$}{
            $I_{t+i} \leftarrow \III[t+i]$
            
            $W'_{t+i+1} \leftarrow \texttt{update}(W'_{t+i}, D[\III[t+i]])$
        }
        
        \If{$d(W'_{t+k}, \WWW[t+k])> \delta$}{
            {\bf return} $\mathbf{fail}$
        }
    }
}
\end{algorithm}

Algorithm~\ref{alg:verification} shows the PoL verification process.
$\Verif$ first checks if $W_0$ was sampled from the required distribution using a statistical test (Line 1).
Once every epoch, $\Verif$ records the distances between each two neighboring models in $\mathit{mag}$ (line 7-9);
sort mag to find $Q$ largest distances and verify the corresponding models and data samples via \texttt{verifyEpoch} (Line 12-13). 
Notice that there are at most $\left\lfloor \frac{S}{k}  \right\rfloor$ distances in each epoch, hence $Q \leq \left\lfloor \frac{S}{k}  \right\rfloor$.

In the \texttt{verifyEpoch} function, $\Verif$ first loads the batch of indexes corresponding to the data points used to update the model from $W_t$ to $W_{t+k}$.
Then, it attempts to reproduce $W_{t+k}$ by performing a series of $k$ updates to arrive at $W'_{t+k}$. 
Notice that $W'_{t+k} \neq W_{t+k}$ due to the noise arising from the hardware and low-level libraries such as cuDNN~\cite{chetlur2014cudnn}.
The reproduction error for the $t$-th model is defined as:
\begin{center}
    $\varepsilon_{\mathit{repr}}(t) = d(W_{t+k}, W'_{t+k})$,
\end{center}
where $d$ represents a distance that could be $l_1$, $l_2$, $l_{\infty}$ or $cos$.
It is required that:
\begin{center}
    $\mathrm{max}_t(\varepsilon_{\mathit{repr}}(t)) \ll d_{\mathit{ref}}$,
\end{center}
where $d_{\mathit{ref}} = d(W_T^1, W_T^2)$ is the distance between two models $W_T^1$ and $W_T^2$ trained with the same architecture, dataset, and initialization strategy, but with  different batching strategies and potentially different initial model weights.
A {\em verification threshold} $\delta$ that satisfies:
\begin{center}
   $\mathrm{max}_t(\varepsilon_{\mathit{repr}}(t)) < \delta < d_{\mathit{ref}}$,
\end{center}
should be calibrated before verification starts.
In their experiments, Jia et al.~\cite{PoL} adopted a normalized reproduction error:
\begin{center}
    $||\varepsilon_{\mathit{repr}}(t)|| = \frac{\mathrm{max}_t(\varepsilon_{\mathit{repr}}(t))}{d_{\mathit{ref}}}$
\end{center}
to evaluate the reproducibility.



In the end, we remark that the number of steps $T$ in PoL verification (Algorithm~\ref{alg:verification}) could be different from that in PoL creation (Algorithm~\ref{alg:creation}),
because $\Adv$ could come up with either a stochastic spoofing or a structurally correct spoofing. 

\section{Attack Methodology}
\label{sec:attack}

In this section, we describe our attacks in detail. 
All of our attacks are stochastic spoofing:
the PoL proof generated by \revision{the adversary} $\Adv$ is not exactly the same as the one provided by \revision{the prover} $\Prov$ (in particular, with a smaller number of steps $T'$), but can pass the verification. 

\fig~\ref{fig:overview} shows the basic idea of our attacks: the adversary $\Adv$ first generates dummy model weights: $W_0, ..., W_{T-1}$ (serving as $\WWW$); and then generates ``adversarial examples'' (serving as $\III$) for each pair of neighboring models.
An adversarial example is an instance added with small and intentional perturbations so that a machine learning model will make a false prediction on it. 
In a similar way as optimizing an adversarial example, we could make an arbitrarily-chosen date point ``generate'' a given model (we call it {\em adversarial optimization}), hence making $(\WWW, \III)$ pass the verification.

\begin{figure}[h]
    \centering
    \includegraphics[width=0.98\linewidth]{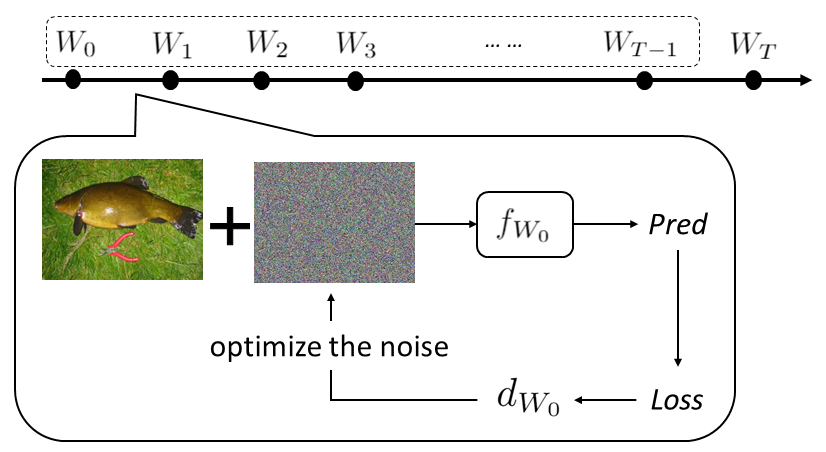}
    \caption{Basic idea of our attacks. 
    The adversary first generates dummy model weights: $W_0, ..., W_{T-1}$ (serving as $\WWW$); and then generates ``adversarial examples'' (serving as $\III$) for each pair of neighboring models.}
    \label{fig:overview}
\end{figure}

Recall that one requirement for a spoof to succeed is that $\Adv$ should spend less cost than 
the PoL generation process described in Algorithm~\ref{alg:creation} (which is $T=E\cdot S$ times of \texttt{update}).
Next, we show how we achieve this.

\subsection{Attack~I}

Our first insight is that there is no need to construct an adversarial example for every pair of neighboring models.
Instead, $\Adv$ could simply update the model from $W_0$ to $W_{T-1}$ using original data points, 
and construct an ``adversarial example'' only from $W_{T-1}$ to $W_T$.
In this case, $\Adv$ only needs to construct a single ``adversarial example'' for the whole attacking process.
Furthermore, $\Adv$ could use a smaller number of steps, denoted as $T'$.

\renewcommand{\algorithmiccomment}[1]{$\triangleright$ #1}
\begin{algorithm}[htb]
\caption{Attack~I}
\label{alg:attack1}
\KwIn {$D$, $f_{W_T}$, $\delta$, $\zeta$, $k$, $E$, $S$}
\KwOut {PoL spoof: $\Proof(\Adv, f_{W_T})=(\WWW, \III, \HHH)$ \\
~~~~~~~~~~~updated dataset: $D$}
$\WWW\leftarrow\{\}$
$\III\leftarrow\{\}$
$\HHH\leftarrow\{\}$

$\WWW.\texttt{append}(\texttt{init}(\zeta))$  \hfill\CommentSty{initialize and append $W_0$}

\For(\hfill\CommentSty{$T' \mod k =0$}){$t = 1  \to T' $}{
    
    $\III.\texttt{append}(\texttt{getBatch}(D))$
   
   \eIf{$t<T'$}{
        $W_{t} \leftarrow \texttt{update}(W_{t-1}, D[\III[t-1]])$
        
        \eIf{$t\mod k = 0$}{
            $\WWW.\texttt{append}(W_t)$ 
        }{
            $\WWW.\texttt{append}(\mathbf{nil})$ 
        }
    }{
        $\texttt{updateDataPoints}(W_{t-1}, W_{T})$
    }
    $\HHH.\texttt{append}(h(D[\III[t-1]]))$
}
~
\SetKwFunction{FMain}{updateDataPoints}
\SetKwProg{Fn}{function}{}{end}

\Fn{\FMain{$W_{t-1}$, $W_{t}$}}{
    
    $W'_{t-1} := W_{t-1}$
    
    $(\XX, \yy) \leftarrow D[\III[t-1]]$
        
    $W'_{t} \leftarrow \texttt{update}(W'_{t-1}, (\XX, \yy))$
        
    \While{$d(W'_{t}, W_{t}) > \delta$}{
        $\mathbf{R} \leftarrow \texttt{zeros}$
                
        $\triangledown_{W'_{t-1}} \leftarrow - \frac{\partial}{\partial W'_{t-1}} L(f_{W'_{t-1}}(\XX+\mathbf{R}), \yy)$
                
        $\mathbb{D}_{t-1} \leftarrow d(W'_{t-1}+\eta \triangledown_{W'_{t-1}}, W_{t})$ + $d(\mathbf{R}, 0)$
                
        $\mathbf{R} \leftarrow \mathbf{R} - \eta' \triangledown_R \mathbb{D}_{t-1}$
                
        $W'_{t} \leftarrow \texttt{update}(W'_{t-1}, (\XX+\mathbf{R}, \yy))$
    }
    $D[\III[t-1]] := (\XX+\mathbf{R}, \yy)$
}
\end{algorithm}

Algorithm~\ref{alg:attack1} shows our first attack.
From $W_0$ to $W_{T'-1}$, it works in the same way as PoL creation (cf. Algorithm~\ref{alg:creation}).
For $W_{T'-1}$, the batch of inputs $(\XX, \yy)$ must be manipulated s.t.:



\begin{center}
    $d(W_{T'}, W_T) \leq \delta.$
\end{center}

Line 22-28 show how $\Adv$ manipulates $\XX$.
Specifically, $\Adv$ first initializes a batch of noise $\RR$ as zeros (line 23).
Then, it feeds $(\XX+\RR)$ to $f_{W'_{t-1}}$ and gets the gradients $\triangledown_{W'_{t-1}}$ (line 25).
Next, $\Adv$ optimizes $\RR$ by minimizing the following distance (line 26-27):
\begin{center}
    $\mathbb{D}_{t-1} \leftarrow d(W'_{t-1}+\eta \triangledown_{W'_{t-1}}, W_{t})$ + $d(\mathbf{R}, 0).$
\end{center}
This distance needs to be differentiable so that $\RR$ can be optimized using standard gradient-based methods\footnote{Specificly, we use L-BFGS for adversarial optimization.}. 
Notice that this optimization requires 2nd order derivatives. We assume that $f_{W}$ is twice differentiable, which holds for most modern machine learning models and tasks.

Clearly, the PoL spoof $\Proof(\Adv, f_{W_T})=(\WWW, \III, \HHH)$ generated by Attack~I can pass the verification process described in Algorithm~\ref{alg:verification}.
It requires {$T'$ times of \texttt{update} (Line~21) plus $N$ times of adversarial optimization (Line~23-27)} (where $N$ is the times that the {\bf while} loop runs).
Recall that our focus is stochastic spoofing:
the PoL proof generated by $\Adv$ is not exactly the same as the one provided by $\Prov$, but can pass the verification. 
Therefore, we can use a $T'$ that is much smaller than $T$.
However, $N$ could be large 
and sometimes even cannot converge.
Next, we show how we optimize the attack so that a small $N$ is able to make the adversarial optimization converge.



\subsection{Attack~II}

The intuition for accelerating  adversarial optimization is to sample the intermediate model weights
in a way s.t.:
\begin{center}
    $d(W_t, W_{t-k}) \leq \delta,~\forall~0< t< T$ and $t \mod k =0$,
\end{center}
This brings at least three benefits:
\begin{enumerate}
    \item  ``Adversarial examples'' become  easier to be optimized.
    \item The $k$ batches of ``adversarial examples'' in each checkpointing interval can be optimized together. (We defer to explain this benefit in Attack~III.)
    \item Intermediate model performance can be guaranteed. (Recall that $\Verif$ might check model performance periodically.) 
\end{enumerate}

\renewcommand{\algorithmiccomment}[1]{$\triangleright$ #1}
\begin{algorithm}[htb]
\caption{Attack~II}
\label{alg:attack2}
\small
\KwIn {$D$, $f_{W_T}$, $\delta$, \textcolor{blue}{$\gamma$}, $\zeta$, $k$, $E$, $S$}
\KwOut {PoL spoof: $\Proof(\Adv, f_{W_T})=(\WWW, \III, \HHH)$ \\
~~~~~~~~~~~updated dataset: $D$}
$\WWW\leftarrow\{\}$
$\III\leftarrow\{\}$
$\HHH\leftarrow\{\}$

\textcolor{blue}{$\WWW.\texttt{append}(\texttt{initW}_0(\zeta, W_T))$}  \hfill\CommentSty{initialize and append $W_0$}

\For(\hfill\CommentSty{$T' \mod k =0$}){$t = 1  \to T' $}{
    
    $\III.\texttt{append}(\texttt{getBatch}(D))$
    
    
    \eIf{$t\mod k = 0$}{
        \eIf(\hfill\CommentSty{no need to append $W_T$}){$t < T'$}{
            \textcolor{blue}{~sample $W_{t}$ s.t., $d(W_{t}, W_{t-k})\leq {\delta}$}
            $\WWW.\texttt{append}(W_t)$ 
        }{
            $W_t := W_T$
        }
        
        $\texttt{updateDataPoints}(W_{t-k}, W_{t})$
        
        \For{$i = (t-k) \to (t-1)$}{
            $\HHH.\texttt{append}(h(D[\III[i]]))$
        }
    }{
        $\WWW.\texttt{append}(\mathbf{nil})$ 
    }
   
}

~
\SetKwFunction{FMain}{updateDataPoints}
\SetKwProg{Fn}{function}{}{end}

\Fn{\FMain{$W_{t-k}$, $W_{t}$}}{
    
    $W'_{t-k} := W_{t-k}$
    
    \For{$i = (t-k) \to (t-1)$}{
        $(\mathbf{X}, \mathbf{y}) \leftarrow D[\III[i]]$
        
        $W'_{i+1} \leftarrow \texttt{update}(W'_i, (\mathbf{X}, \mathbf{y}))$
        
        \While(\hfill){\textcolor{blue}{$d(W'_{i+1}, W'_{i}) > \gamma$}}{
            $\mathbf{R} \leftarrow \texttt{zeros}$
            
            $\triangledown_{W'_i} \leftarrow - \frac{\partial}{\partial W'_i} L(f_{W'_i}(\mathbf{X}+\mathbf{R}), \mathbf{y})$
            
            \textcolor{blue}{ 
            $\mathbb{D}_i \leftarrow d(\triangledown_{W'_i} , 0)$ + $d(\mathbf{R}, 0)$
            }
            
            $\mathbf{R} \leftarrow \mathbf{R} - \eta' \triangledown_R \mathbb{D}_i$
            
            $W'_{i+1} \leftarrow \texttt{update}(W'_i, (\mathbf{X}+\mathbf{R}, \mathbf{y}))$
        }
        $D[\III[i]] := (\mathbf{X}+\mathbf{R}, \mathbf{y})$
    }
}
\end{algorithm}

Algorithm~\ref{alg:attack2} shows Attack~II. We highlight the key differences (compared to Attack~I) in blue.

This time, $\Adv$ initializes $W_0$ via $\texttt{initW}_0$ (line 2), which ensures that $W_0$ follows the given distribution $\zeta$,
and minimizes $d(W_0, W_T)$ at the same time. 
It works as follows:
\begin{enumerate}
    \item Suppose there are $n$ elements in $W_T$,
    $\Adv$ puts these elements into a set $S_1$. Then, $\Adv$ samples $n$ elements: $v_1, ..., v_n$ from the given distribution $\zeta$, and puts them into another set $V_2$.
    \item $\Adv$ finds the largest elements $w$ and $v$ from $S_1$ and $S_2$ respectively.
    Then, $\Adv$ puts $v$ into $W_0$ according
    to $w$'s indices in $W_T$.
    \item $\Adv$ removes $(w, v)$ from $(S_1, S_2)$, and repeats step 2) until $S_1$ and $S_2$ are empty.
\end{enumerate}
Our experimental results show that this process can initialize a $W_0$ that meets our requirements. 

For other $W_t$s ($t>0$), $\Adv$ can initialize them by equally dividing the distance between $W_0$ and $W_T$.
If \revision{the number of steps for spoofing (i.e., $T'$)} is large enough (i.e., there are enough $W_t$s), the condition ``$d(W_t, W_{t-k}) \leq \delta$'' can be trivially satisfied.

Another major change in Attack~II is that $\Adv$ optimizes the noise $\RR$ by minimizing the following distance (line 27):
\begin{center}
    $\mathbb{D}_i \leftarrow d(\triangledown_{W'_i} , 0)$ + $d(\mathbf{R}, 0)$,
\end{center}
and the condition for terminating the adversarial optimization is $d(W'_{i+1}, W'_i) >\gamma$ where $\gamma \ll \delta$ (Line 25).
This guarantees that the model is still close to itself after a single step of update.
Since the distance between $W_{t-k}$ and $W_{t}$ is smaller than $\delta$ after initialization, after $k$ steps of updates, their distance is still smaller than $\delta$: $d(W_t, W'_t) < \delta$.

Interestingly, this change makes the adversarial optimization become easier to converge. 
Recall that in Attack~I, $\Adv$ has to adjust the loss function $L(f_{W'_i}(\mathbf{X}+\mathbf{R}), \mathbf{y})$ to minimize $d(W'_{t-1}+\eta \triangledown_{W'_{t-1}}, W_{t})$. 
This is difficult to achieve because gradient-based training is used to minimize (not adjust) 
the loss function.
Thanks to the new $\mathbb{D}_i$, $\Adv$ can simply minimize the loss function in Attack~II.
In another word, the adversarial optimization process in Attack~II is more close to normal training.
Table~\ref{tab:loss-diff} shows that on CIFAR-10, after 10 steps of adversarial optimization, the loss function decreases from 0.43 to 0.04, and the gradients decrease from 61.13 to 0.12.
Both are small enough to pass the verification.
That is to say, \revision{the number of while loops} $N$ can be as small as 10 in Attack~II.


\begin{table}[htb]
\centering
\caption{The changes of loss and the gradients after 20 steps of adversarial optimization on CIFAR-10}
\begin{tabular}{@{}lll@{}}
\toprule
 & $L(f_{W'_i}(\mathbf{X}+\mathbf{R}), \mathbf{y})$ & $\left \|\triangledown_{W'_i}\right \|^{2}$ \\ \midrule
Before & 0.43 $\pm$ 0.18 &  61.13 $\pm$ 45.86 \\
After & 0.04 $\pm$ 0.01 & {0.12 $\pm$ 0.05}\\ \bottomrule
\end{tabular}
\label{tab:loss-diff}
\end{table}

\revision{\Paragraph{Cost comparison.}
\textlabel{It is easy to see that Attack~II requires}{e:3:1} {$T'$ times of \texttt{update} (Line~24) plus $T' \cdot N$ times of adversarial optimization (Line~26-30)}, where $N=10$.
Each \texttt{update} requires one gradient computation and each adversarial optimization requires three gradient computations.
In total, {Attack~II requires $31T'$ gradient computations}.
As a comparison, generating a PoL proof requires $T= E\cdot S$ gradient computations.
In~\cite{PoL}, they set $E=200$ and $S=390$, hence $T = 78,000$.
Therefore, the cost for Attack~II is smaller than the trainer’s cost as long as $T' < \frac{78,000}{31} \approx 2,516$. 
Our experimental results show that this is more than enough for the spoof to pass the verification (cf. Section~\ref{sec:eval}). 
Next, we show how we further optimize our attack.
}

\subsection{Attack~III}

\renewcommand{\algorithmiccomment}[1]{$\triangleright$ #1}
\begin{algorithm}[htb]
\caption{Attack~III}
\label{alg:attack3}
\small
\KwIn {$D$, $f_{W_T}$, $\delta$, $\gamma$, $\zeta$, $k$, $E$, $S$}
\KwOut {PoL spoof: $\Proof(\Adv, f_{W_T})=(\WWW, \III, \HHH)$ \\
~~~~~~~~~~~updated dataset: $D$}
$\WWW\leftarrow\{\}$
$\III\leftarrow\{\}$
$\HHH\leftarrow\{\}$

$\WWW.\texttt{append}(\texttt{initW}_0(\zeta, W_T))$  \hfill\CommentSty{initialize and append $W_0$}

\For(\hfill\CommentSty{$T' \mod k =0$}){$t = 1  \to T' $}{
    
    $\III.\texttt{append}(\texttt{getBatch}(D))$
    
    \eIf{$t\mod k = 0$}{
        \eIf(\hfill\CommentSty{no need to append $W_T$}){$t < T$}{
            ~sample $W_{t}$ s.t., {$d(W_{t}, W_{t-k})\leq {\delta}$}
            $\WWW.\texttt{append}(W_t)$ 
        }{
            $W_t := W_T$
        }
        
        $\texttt{updateDataPoints}(W_{t-k}, W_{t})$
        
        \For{$i = (t-k) \to (t-1)$}{
            $\HHH.\texttt{append}(h(D[\III[i]]))$
        }
    }{
        $\WWW.\texttt{append}(\mathbf{nil})$ 
    }
}

~
\SetKwFunction{FMain}{updateDataPoints}
\SetKwProg{Fn}{function}{}{end}

\Fn{\FMain{$W_{t-k}$, $W_{t}$}}{
    
    
        \textcolor{blue}{ 
        $(\mathbf{X}, \mathbf{y}) \leftarrow [D[\III[t-k]]...D[\III[t-1]]]$
        }
        
        $W'_{t} \leftarrow \texttt{update}(W_{t-k}, (\mathbf{X}, \mathbf{y}))$
               
        \While{$d(W'_{t}, W_{t}) > \gamma \textcolor{blue}{-\sigma}$}{
            
            $\mathbf{R} \leftarrow \texttt{zeros}$

            \textcolor{blue}{ 
            $\triangledown_{W_{t-k}} \leftarrow - \frac{\partial}{\partial W_{t-k}} L(f_{W_{t-k}}(\mathbf{X}+\mathbf{R}), \mathbf{y})$
            }
            
            \textcolor{blue}{ 
            $\mathbb{D}_{t-k} \leftarrow d(\triangledown_{W_{t-k}} , 0)$ + $d(\mathbf{R}, 0)$
            }
            
            $\mathbf{R} \leftarrow \mathbf{R} - \eta' \triangledown_R \mathbb{D}_{t-k}$
            
            \textcolor{blue}{ 
            $W'_{t} \leftarrow \texttt{update}(W_{t-k}, (\mathbf{X}+\mathbf{R}, \mathbf{y}))$
            }
        }
        \textcolor{blue}{ 
        $[D[\III[t-k]]...D[\III[t-1]]] := (\mathbf{X}+\mathbf{R}, \mathbf{y})$}
}
\end{algorithm}

Algorithm~\ref{alg:attack3} shows Attack III. Again, we highlight the key differences (compared to Attack II) in blue. 
The major change is that $\Adv$ optimizes all $k$ batches of data points  together in \texttt{updateDataPoints}.
The distance between $W'_t$ and $W_t$ should be limited to $\gamma-\sigma$ (where $0<\sigma<\gamma$), instead of $\gamma$. 
We will show the reason in the proof of Corollary~\ref{coro:proof}.

Optimizing all $k$ batches together reduces the complexity to 
{$\frac{T'}{k}$ times of \texttt{update} (Line~22) plus $\frac{T'}{k}\cdot N$ times of adversarial optimization (Line 24-28)}.
At first glance, this will not pass the verification because $\Verif$ will run \texttt{update} for each batch individually. 
\revision{
\textlabel{However, since}{r1:2:1} $W_{t-k}, ..., W_{t-1}$ are all very close to each other, 
taking a gradient with respect to $W_{t-k+j}$ is similar to taking a gradient with respect to $W_{t-k}$. 
Consequently, we can optimize all $k$ batches directly with respect to $W_{t-k}$.
}
The gap only depends on $k$, hence we can make a trade-off. 
Next, we formally prove this argument.  

\begin{corollary}
\label{coro:proof}
Let $(W_{t-k}, W_t)$ be an input to \texttt{updateDataPoints} in Attack III.
Let $\{\hat{W}_{t-k-1},...,\hat{W}_{t}\}$ be the model weights computed by $\Verif$ based on $W_{t-k}$ during PoL verification. 
Assuming the loss function $L(f_{W}(\mathbf{X}),\mathbf{y}) \in C^2(\Omega)$, where $\Omega$ is a closed, convex and connected subset in $\mathbb{R}^n$, and $\{\hat{W}_{t-k-1},...,\hat{W}_{t}\} \in \Omega$. Then, 
$$||\hat{W}_t-W_t||\leq \eta^2\alpha\beta \frac{(k-1)(k-2)}{2} + \gamma - \sigma,$$
where $\alpha$ and $\beta$ are the upper bounds of first and second order derivative\footnote{Empirically, $\alpha$ is 0.03 in average and $\beta$ is 0.025 in average.} of $L(f_{W}(\mathbf{X}),\mathbf{y})$.
\end{corollary}

\begin{proof}
Let $\XX=[\xx_1, \xx_2, ..., \xx_k]^T$ be the $k$ batches used to update $W_{t-k}$. 
Denote 
$$L_i(W) =L(f_W(\xx_i),\yy_i)\in C^2(\Omega),$$
$$\triangledown_i(W) = \frac{\partial}{\partial W}L_i\in C^1(\Omega),$$
$$\triangledown'_i(W) = \frac{\partial^2}{\partial W^2}L_i\in C^0(\Omega).$$
Then, $||\triangledown_i(W)||<\alpha$ and $||\triangledown'_i(W)||<\beta$.

In Attack III, (Line 22 of Algorithm~\ref{alg:attack3}), $W'_t$ is calculated as 
\begin{equation*}
    \begin{aligned}
    W'_t =&~W_{t-k}-\frac{\eta''}{k}(
       \triangledown_1(W_{t-k})+\triangledown_2(W_{t-k})+...+\triangledown_k(W_{t-k})
    )
    \end{aligned}
\end{equation*}
Whereas, in PoL verification (Line 29 of Algorithm~\ref{alg:verification}),
\begin{equation*}
    \begin{aligned}
    \hat{W}_{t-k+1} =&~W_{t-k}-\eta\triangledown_1(W_{t-k}) \\
    \hat{W}_{t-k+2} =&~\hat{W}_{t-k+1}-\eta\triangledown_2(\hat{W}_{t-k+1}) \\
    ...\\
    \hat{W}_t =&~\hat{W}_{t-1}-\eta\triangledown_k(\hat{W}_{t-1})
    \end{aligned}
\end{equation*}
It is identical to
\begin{equation*}
    \begin{aligned}
    \hat{W}_t = &W_{t-k}-\eta(
     \triangledown_1(W_{t-k})+\triangledown_2(\hat{W}_{t-k+1})+...+\triangledown_k(\hat{W}_{t-1})
       )
    \end{aligned}
\end{equation*}
If $\Adv$ sets $\eta'' = k\eta$, then
\begin{equation*}
    \begin{aligned}
       \hat{W}_t -W'_t = \eta[&(\triangledown_2(W_{t-k})-\triangledown_2(\hat{W}_{t-k+1})+\\&
       (\triangledown_3(W_{t-k})-\triangledown_3(\hat{W}_{t-k+2})+...+\\
       &(\triangledown_k(W_{t-k})-\triangledown_k(\hat{W}_{t-1})]
    \end{aligned}
\end{equation*}
Assuming $[\hat{W}_{t-k+l},W_{t-k}]=\{W\in \mathbb{R}^n,W=\hat{W}_{t-k+l}+\theta h,0\leq \theta \leq 1\}$ is a closed set, and $\triangledown_i(W)\in C^1(\Omega)$. 
Based on the finite-increment theorem~\cite{book-1469653}, we have
    \begin{equation*}
    \begin{aligned}
||\triangledown_i(W_{t-k}) - \triangledown_i(\hat{W}_{t-k+l})||&\leq \sup_{W}||\frac{\partial\triangledown_i(W)}{\partial W}||\cdot||h|| \\&\leq  \beta||W_{t-k}-\hat{W}_{t-k+l}||
        \end{aligned}
   \end{equation*}
Given 
    \begin{equation*}
    \begin{aligned}
||W_{t-k}-\hat{W}_{t-k+l}|| &=\eta ||\triangledown_1(W_{t-k}) + \triangledown_2(\hat{W}_{t-k+1}) + ... \\ &~~~~+  \triangledown_{l-1}(\hat{W}_{t-k+l-1})|| \\ 
&\leq (l-1) \eta \alpha,
        \end{aligned}
   \end{equation*}
we have
$$||\triangledown_i(W_{t-k}) - \triangledown_i(\hat{W}_{t-k+l})|| \leq 
(l-1)\eta \alpha \beta.$$
Then,
\begin{equation*}
\begin{aligned}
          \hat{W}_t -W'_t  &\leq 
          \eta^2\alpha\beta\sum_{l=1}^{k-1}{(l-1)} \\
          &= \eta^2\alpha\beta \frac{(k-1)(k-2)}{2}
\end{aligned}
\end{equation*}

Recall that $d(W'_t,W_t) \leq \gamma - \sigma$ (Line 23 in Algorithm~\ref{alg:attack3}). 
Then, 
\begin{equation*}
\begin{aligned}
    ||\hat{W}_t-W_t|| &= ||\hat{W}_t-W'_t+W'_t-W_t||\\
    &\leq ||\hat{W}_t-W'_t|| + ||W'_t-W_t||\\
    &\leq \eta^2\alpha\beta \frac{(k-1)(k-2)}{2} + \gamma - \sigma
    \end{aligned}
\end{equation*}


\end{proof}

Therefore, Attack III can pass the verification if we set $\sigma > \eta^2\alpha\beta \frac{(k-1)(k-2)}{2}$.

\revision{\Paragraph{Cost comparison.}
\textlabel{Recall that the complexity for Attack~III is}{e:3:2} $\frac{T'}{k}$ times of \texttt{update} plus $\frac{T'}{k} \cdot N$ times of adversarial optimization, where $N=10$;
each \texttt{update} requires one gradient computation and each adversarial optimization requires three gradient computations.
In total, {Attack~III requires $31\frac{T'}{k}$ gradient computations}.
Given that generating a PoL proof requires requires $T = 78,000$ gradient computations,
the cost for Attack~III is smaller than the trainer’s cost as long as $\frac{T'}{k} < \frac{78,000}{31} \approx 2,516$. 
In our experiments, we show Attack~III can pass the verification when we set $k=100$ (cf. Section~\ref{sec:eval}).
Then, $T'<251,600$ is the condition for the cost of Attack~III to be smaller than the trainer’s cost, making the parameter setting more flexible than Attack~II.
Notice that our experiments did not show such good results, because $k$ batches of samples  exceeds the memory size and we have to load them separately.
Nevertheless, the results of Attack~III are still much better than Attack~II.
}

\begin{figure*}[ht]
    \centering
    \subfigure[Normalized reproduction error in $l_1$]{
    \includegraphics[width=0.31\linewidth]{ 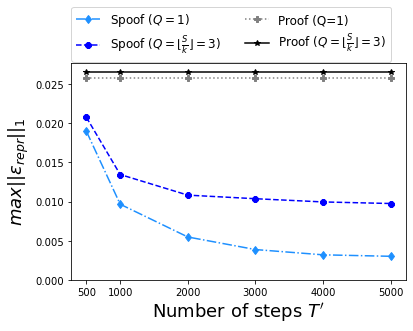}\label{fig:2.1}}
    \subfigure[Normalized reproduction error in $l_2$]{
    \includegraphics[width=0.31\linewidth]{ 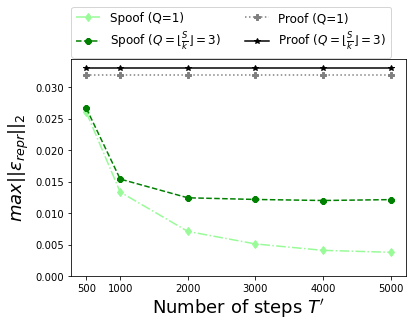}\label{fig:2.2}}
    \subfigure[Normalized reproduction error in $l_{\infty}$]{
    \includegraphics[width=0.31\linewidth]{ 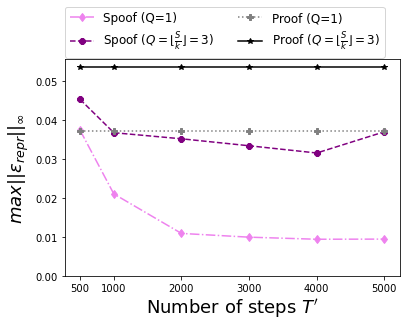}\label{fig:2.3}}
    \\
    \subfigure[Normalized reproduction error in $cos$]{
    \includegraphics[width=0.31\linewidth]{ 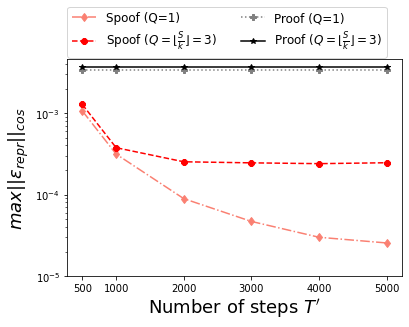}\label{fig:2.4}}
    \subfigure[Spoof generation time.]{
    \includegraphics[width=0.31\linewidth]{ 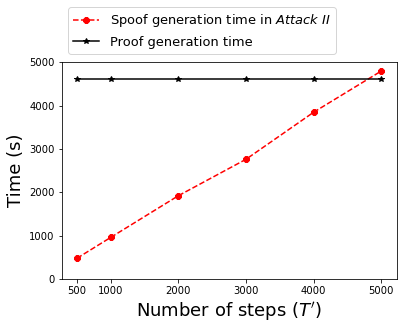}\label{fig:2.5}}
    \subfigure[Spoof size.]{
    \includegraphics[width=0.31\linewidth]{ 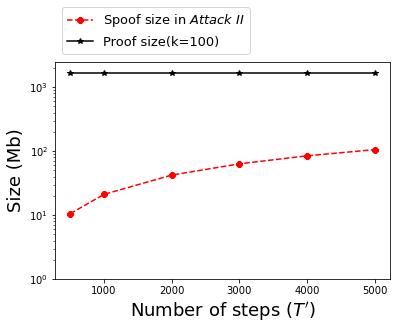}\label{fig:2.6}} 
   \caption{Attack II on CIFAR-10}
    \label{fig:Attack2_cifar10_eps}
\end{figure*}

\begin{figure*}
    \centering
        \subfigure[Normalized reproduction error in $l_1$]{
    \includegraphics[width=0.31\linewidth]{ 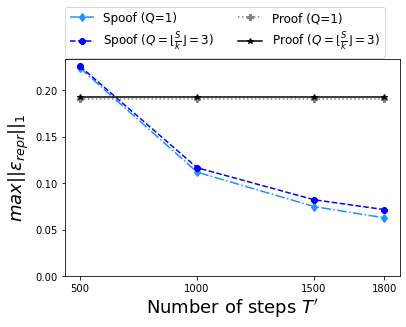}\label{fig:3.1}}
    \subfigure[Normalized reproduction error in $l_2$]{
    \includegraphics[width=0.31\linewidth]{ 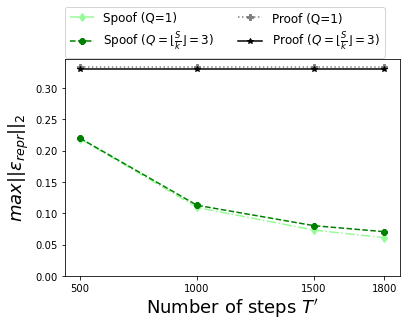}\label{fig:3.2}}
    \subfigure[Normalized reproduction error in $l_{\infty}$]{
    \includegraphics[width=0.31\linewidth]{ 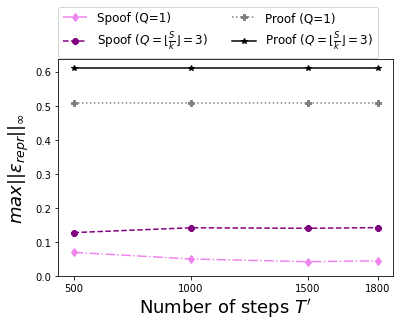}\label{fig:3.3}}
    \\
    \subfigure[Normalized reproduction error in $cos$]{
    \includegraphics[width=0.31\linewidth]{ 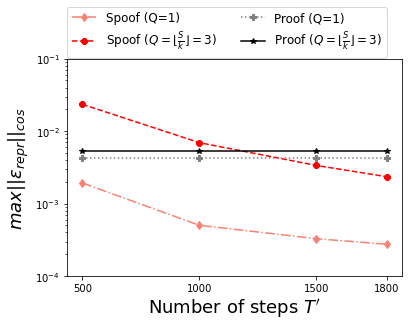}\label{fig:3.4}}
    \subfigure[Spoof generation time.]{
    \includegraphics[width=0.31\linewidth]{ 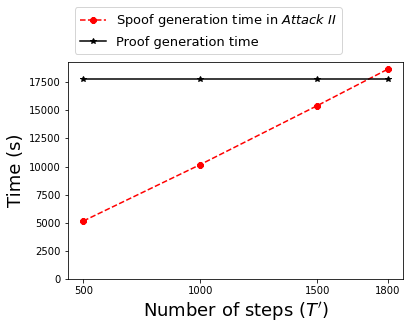}\label{fig:3.5}}
    \subfigure[Spoof size.]{
    \includegraphics[width=0.31\linewidth]{ 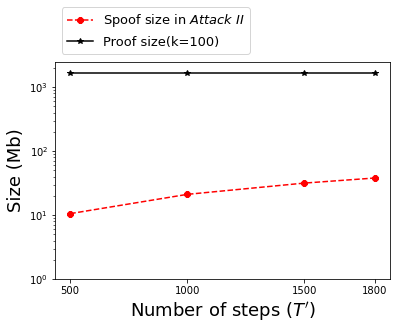}\label{fig:3.6}}    
   \caption{Attack II on CIFAR-100}
    \label{fig:Attack2_cifar100_eps}
\end{figure*}

\begin{figure*}
    \centering
    \subfigure[Normalized reproduction error in $l_1$]{
    \includegraphics[width=0.31\linewidth]{ 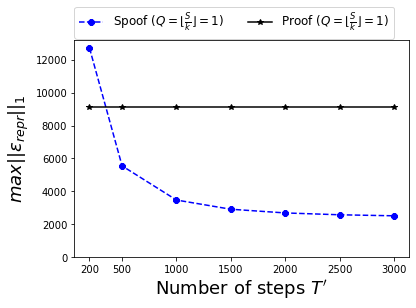}\label{fig:4.1}}
    \subfigure[Normalized reproduction error in $l_2$]{
    \includegraphics[width=0.31\linewidth]{ 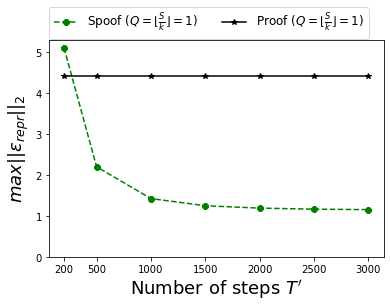}\label{fig:4.2}}
    \subfigure[Normalized reproduction error in $l_{\infty}$]{
    \includegraphics[width=0.31\linewidth]{ 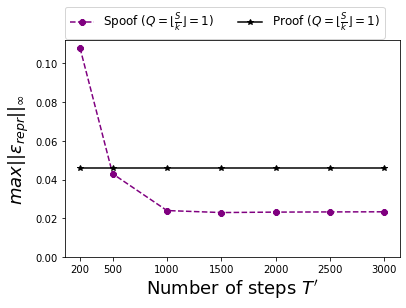}\label{fig:4.3}}
    \\
    \subfigure[Normalized reproduction error in $cos$]{
    \includegraphics[width=0.31\linewidth]{ 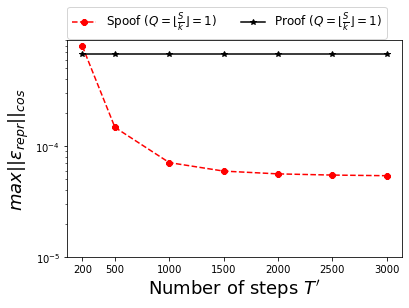}\label{fig:4.4}}
    \subfigure[Spoof generation time.]{
    \includegraphics[width=0.31\linewidth]{ 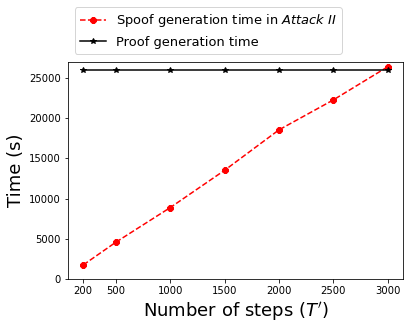}\label{fig:4.5}}
    \subfigure[Spoof size.]{
    \includegraphics[width=0.31\linewidth]{ 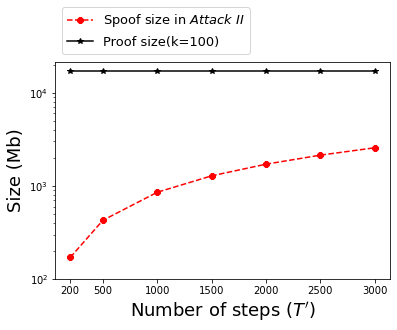}\label{fig:4.6}}        
   \caption{Attack II on ImageNet}
    \label{fig:Attack2_imagenet_eps}
\end{figure*}


\section{Evaluation}
\label{sec:eval}

In this section, we 
evaluate our attacks in two metrics:
\begin{itemize}
    \item {\bf Reproducibility.} We need to show that the normalized reproduction errors (cf. Section~\ref{sec:verification}) in $l_1$, $l_2$, $l_{\infty}$ and $cos$ introduced by our PoL spoof are smaller than  those introduced by a legitimate PoL proof.
    That means, as long as \deleted{$\Proof(\Prov, f_{W_T})$}\revision{the PoL proof} can pass the verification, \deleted{$\Proof(\Adv, f_{W_T})$}\revision{our spoof} can pass the verification as well.
    \item {\bf Spoof cost.} Recall that a successful PoL spoof requires the attacker to spend less computation and storage than the prover (cf. Section~\ref{sec:definition}).
    Therefore, we need to show that the generation time and size of \deleted{$\Proof(\Adv, f_{W_T})$}\revision{the spoof} are smaller than those of \deleted{$\Proof(\Prov, f_{W_T})$}\revision{the proof}.
    
\end{itemize}

\subsection{Setting}

Following the experimental setup in~\cite{PoL},
we evaluate our attacks for ResNet-20~\cite{He2016resnet} and ResNet-50~\cite{He2016resnet} on CIFAR-10~\cite{krizhevsky2009learning} and CIFAR-100~\cite{krizhevsky2009learning} respectively. 
Each of them contains 50,000 training images and 10,000 testing images; and each image is of size 32×32×3. 
CIFAR-10 only has 10 classes and CIFAR-100 has 100 classes.

We reproduced the results in~\cite{PoL} as baselines for our attacks. 
Namely, we generate \deleted{$\Proof(\Prov, f_{W_T})$}\revision{the PoL proof} by training both models for 200 epochs with 390 steps in each epoch (i.e., $E=200$, $S=390$) with batch sizes being 128.
The authors in~\cite{PoL} suggest to set \revision{the number of steps ($k$) in each checkpointing interval} equal to $S$, i.e., $k=S$, but in that case \revision{the verifier can only verify one model per epoch}, i.e., $Q=1$ (cf. Algorithm~\ref{alg:verification}).
Therefore, we set  $k$ as 100.
\revision{\textlabel{Under this setting, the models achieve}{e:1:1} 90.13\% accuracy on CIFAR-10 and 77.03\% accuracy on CIFAR-100.}
All reproduced results are consistent with the results reported in~\cite{PoL}. 

Since our attacks are stochastic spoofing, where \deleted{$\Proof(\Adv, f_{W_T})$}\revision{the spoof} is {\em not} required to be exactly the same as \deleted{$\Proof(\Prov, f_{W_T})$}\revision{the proof}, we could use different $T$, $k$ and batch size in our attacks.
\revision{\textlabel{Nevertheless}{r2:3:4}, we still set $k$ as 100 and batch size as 128 for all of our experiments, to show that our attacks are still valid even in a tough hyperparameter setting.}

\revision{
\textlabel{
Recall that}{e:5:1} $\Verif$ only verifies
$Q < \frac{S}{k}$ largest updates for each epoch (cf. Algorithm~\ref{alg:verification}).
Given that $S=390$ and $k=100$, we have $Q \leq 3$.
In~\cite{PoL}, Jia et al. claim that $Q=1$ would be sufficient for the verifier to detect spoofing, and they
use $Q=1$ for their experiments. 
To be cautious, we 
run experiments with both $Q = \left\lfloor \frac{S}{k}  \right\rfloor = 3$ and $Q=1$.
}

\textlabel{To show that our attacks can scale to}{r2:3:1} more complex datasets, we evaluate our attacks on ImageNet~\cite{deng2009imagenet}.
We pick the first 10 classes from ImageNet;  it contains 13,000 training images and 500 testing images, and each image is resized as 128$\times$128.
We train ResNet-18~\cite{He2016resnet} with $E=200$, $S=101$, $k=100$ and batch size being 128; it achieves 71\% accuracy.
In this case, $Q$ can only be one.
We again set $k=100$ and batch size as 128 for our attacks on ImageNet.

\revision{
\textlabel{
All experiments}{e:4:1} were repeated 5 times with different  random initializations, and the
averages are reported.}

\subsection{results}

\Paragraph{Attack I.}
\revision{
Table~\ref{tab:attack~I} \textlabel{shows}{r1:3:1} the evaluation results of Attack~I on CIFAR-10.
It shows that after a similar amount of training time, the normalized reproduction errors introduced by PoL spoof are significantly larger than those introduced by PoL proof.
This is unsurprising since Attack~I is just a strawman attack.
Next, we focus on evaluating our real attacks: Attack~II and Attack~III.
}

\begin{table}[htb]
\centering
\caption{Attack~I on CIFAR-10 and CIFAR-100}
\begin{tabular}{c|ll|ll}
\multicolumn{1}{l|}{} & \multicolumn{2}{c|}{CIAFR-10} & \multicolumn{2}{c}{CIFAR-100} \\ \hline
                & Spoof         & Proof         & Spoof          & Proof         \\
$l_1$           &  0.2152       & 0.0265        &  0.7558       & 0.1911                   \\
$l_2$           &  0.2796       & 0.0333        &  0.6552       & 0.3330               \\
$l_{\infty}$    &  0.2291       & 0.0283        & 0.1518        & 0.5075               \\
$cos$           &  0.0758       & 0.0038        & 0.1345        & 0.0043               \\
{Time (s)}      &  4,591         & 4,607          & 18,307         & 17,756              
\end{tabular}
\label{tab:attack~I}
\end{table}

\Paragraph{Attack II.}
\fig~\ref{fig:Attack2_cifar10_eps} shows the evaluation results of Attack II on CIFAR-10. 
The results show that the normalized reproduction errors introduced by \deleted{$\Proof(\Adv, f_{W_T})$}\revision{the PoL spoof} are always smaller than those introduced by \deleted{$\Proof(\Prov, f_{W_T})$}\revision{the PoL proof} in $l_1$, $l_2$ and $l_{cos}$ (\fig~\ref{fig:2.1}, \ref{fig:2.2} and \ref{fig:2.4}).
For $l_{\infty}$, it requires $T'>1,000$ for \deleted{$\Proof(\Adv, f_{W_T})$}\revision{the spoof} to be able to pass the verification  (\fig~\ref{fig:2.3}).
On the other hand, when $T'>4,000$, the generation time of \deleted{$\Proof(\Adv, f_{W_T})$}\revision{the spoof} is larger than that of \deleted{$\Proof(\Prov, f_{W_T})$}\revision{the proof} (\fig~\ref{fig:2.5}).
That means {$1,000<T'<4,000$ would be the condition for Attack II to succeed on CIFAR-10}.
Notice that the spoof size is always smaller than the proof size.


\fig~\ref{fig:Attack2_cifar100_eps} shows that {$1,300<T'<1,700$ is the condition for Attack II to be successful on CIFAR-100}.
\revision{
\fig~\ref{fig:Attack2_imagenet_eps} \textlabel{shows}{r2:3:2} that {$500<T'<3,000$ is the condition for Attack II to succeed on ImageNet}.
}

\begin{figure*}[ht]
    \centering
    \subfigure[Normalized reproduction error in $l_1$.]{
    \includegraphics[width=0.31\linewidth]{ 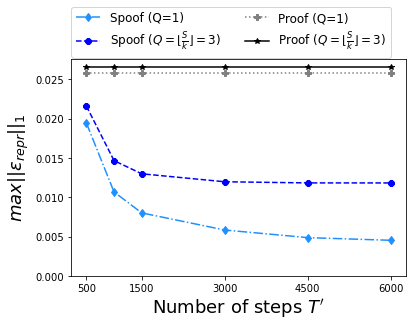}\label{fig:6.1}}
    \subfigure[Normalized reproduction error in $l_2$.]{
    \includegraphics[width=0.31\linewidth]{ 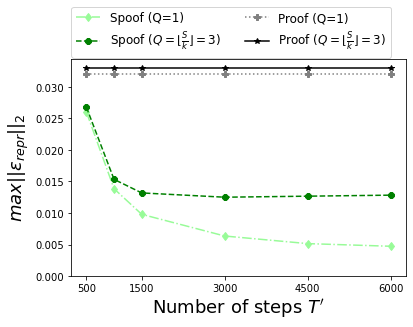}\label{fig:6.2}}
    \subfigure[Normalized reproduction error in $l_{\infty}$.]{
    \includegraphics[width=0.31\linewidth]{ 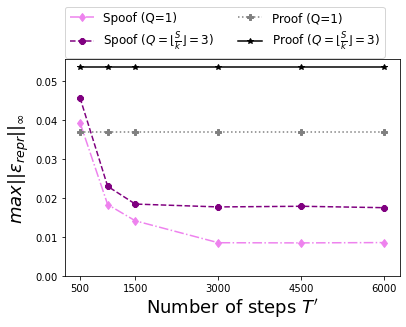}\label{fig:6.3}}
    \subfigure[Normalized reproduction error in $cos$.]{
    \includegraphics[width=0.31\linewidth]{ 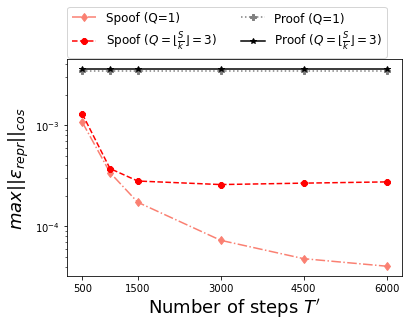}\label{fig:6.4}}
    \subfigure[Spoof generation time.]{
    \includegraphics[width=0.31\linewidth]{ 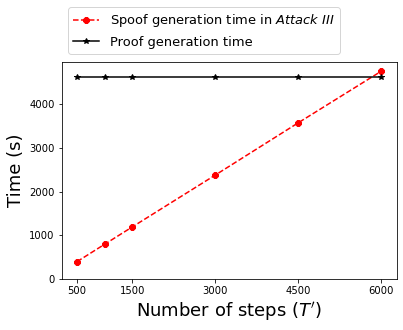}\label{fig:6.5}}
    \subfigure[Spoof Size.]{
    \includegraphics[width=0.31\linewidth]{ 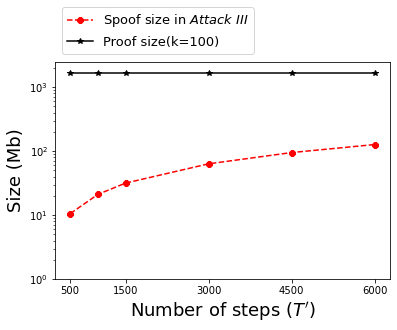}\label{fig:6.6}}
    
    \caption{Attack III on CIFAR-10.}
    \label{fig:Attack3_cifar10_eps}
\end{figure*}

\begin{figure*}
    \centering
        \subfigure[Normalized reproduction error in $l_1$.]{
    \includegraphics[width=0.31\linewidth]{ 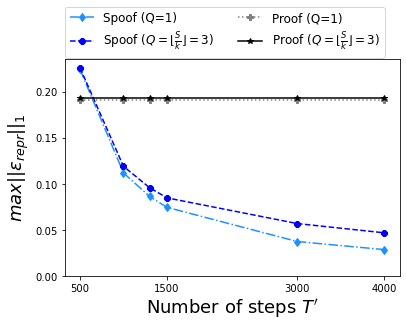}\label{fig:7.1}}
    \subfigure[Normalized reproduction error in $l_2$.]{
    \includegraphics[width=0.31\linewidth]{ 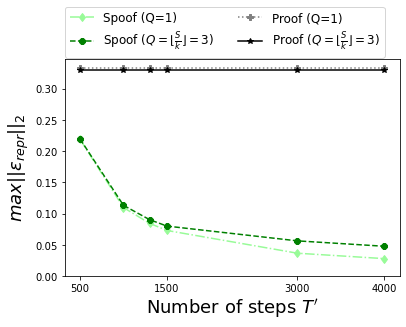}\label{fig:7.2}}
    \subfigure[Normalized reproduction error in $l_{\infty}$.]{
    \includegraphics[width=0.31\linewidth]{ 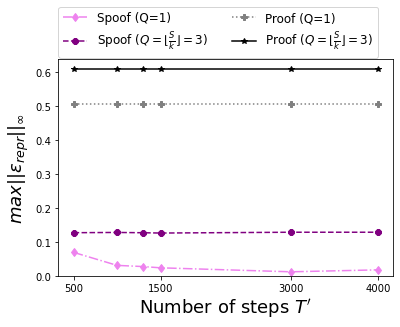}\label{fig:7.3}}
    \subfigure[Normalized reproduction error in $cos$.]{
    \includegraphics[width=0.31\linewidth]{ 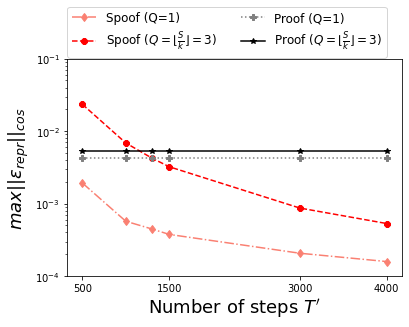}\label{fig:7.4}}
    \subfigure[Spoof generation time.]{
    \includegraphics[width=0.31\linewidth]{ 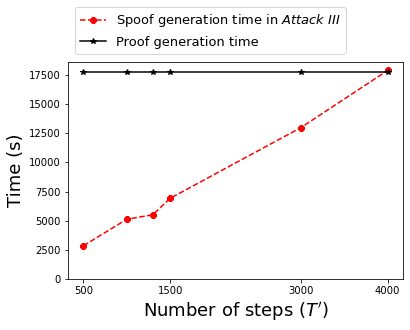}\label{fig:7.5}}
    \subfigure[Spoof Size.]{
    \includegraphics[width=0.31\linewidth]{ 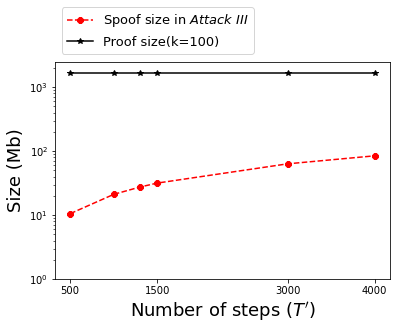}\label{fig:7.6}}
    \caption{Attack III on CIFAR-100. }
    \label{fig:Attack3_cifar100_eps}
\end{figure*}

\begin{figure*}
    \centering
    \subfigure[Normalized reproduction error in $l_1$.]{
    \includegraphics[width=0.31\linewidth]{ 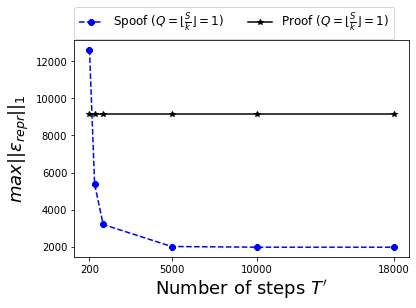}\label{fig:8.1}}
    \subfigure[Normalized reproduction error in $l_2$.]{
    \includegraphics[width=0.31\linewidth]{ 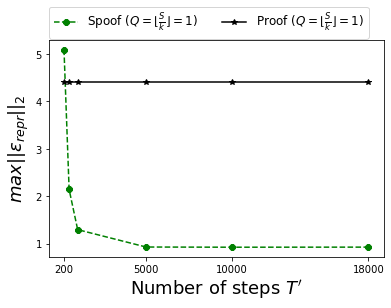}\label{fig:8.2}}
    \subfigure[Normalized reproduction error in $l_{\infty}$.]{
    \includegraphics[width=0.31\linewidth]{ 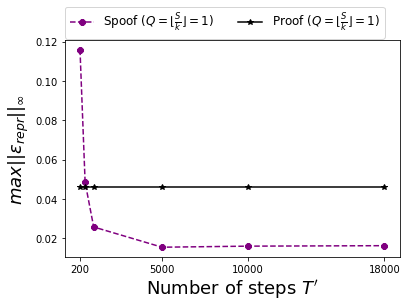}\label{fig:8.3}}
    \subfigure[Normalized reproduction error in $cos$.]{
    \includegraphics[width=0.31\linewidth]{ 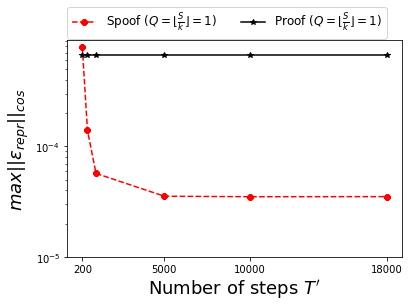}\label{fig:8.4}}
    \subfigure[Spoof generation time.]{
    \includegraphics[width=0.31\linewidth]{ 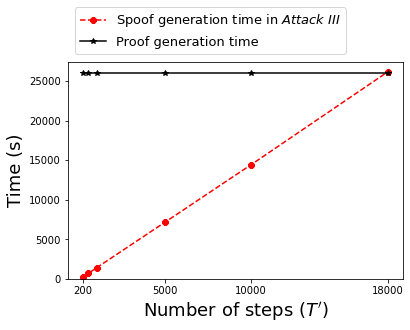}\label{fig:8.5}}
    \subfigure[Spoof Size.]{
    \includegraphics[width=0.31\linewidth]{ 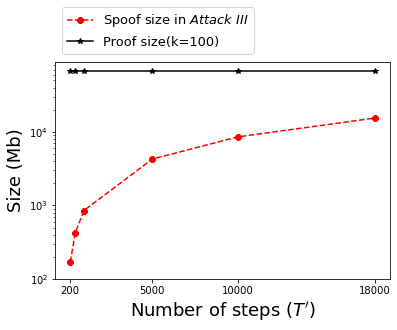}\label{fig:8.6}}
    
    \caption{Attack III on ImageNet.}
    \label{fig:Attack3_imagenet_eps}
\end{figure*}

\Paragraph{Attack III.}
\fig~\ref{fig:Attack3_cifar10_eps} shows the results of Attack III on CIFAR-10. 
The results show that the normalized reproduction errors introduced by \deleted{$\Proof(\Adv, f_{W_T})$}\revision{the spoof} are always smaller than those introduced by \deleted{$\Proof(\Prov, f_{W_T})$}\revision{the proof} in all 4 distances when $T'>1,000$.
In terms of spoof generation time, the number of steps $T'$ can be as large as 6,000.
That means {the condition for Attack III to be successful on CIFAR-10 is $1,000<T'<6,000$}.


\fig~\ref{fig:Attack3_cifar100_eps} shows that {the condition for Attack III to succeed in CIFAR-100 is $1,500<T'<4,000$.}
\revision{
\fig~\ref{fig:Attack3_imagenet_eps} \textlabel{shows}{r2:3:3} that {$500<T'<18,000$ is the condition for Attack III to be successful on ImageNet}.
}


In~\cite{PoL}, the authors suggest to check model performance periodically.
Therefore, we need to make sure that the performance of the intermediate models generated by our attacks follow the same tend as those in  \deleted{$\Proof(\Prov, f_{W_T})$}\revision{the proof}.
We can achieve this by adjusting the extent of perturbations on $W_T$ (Line 7 in Algorithm~\ref{alg:attack3}).
Specifically, we can add a large extent of perturbations when $T'$ is small and add a small extent of perturbations when $T'$ is large.
\fig~\ref{fig:intermediate_model_accuracyl} (in Appendix) shows the model performance in both CIFAR-10 and CIFAR-100.
The $x$-axis presents the progress of training. 
For example when $x=0.2$, the corresponding $y$ represents the $0.2\cdot T$-th model performance in \deleted{$\Proof(\Prov, f_{W_T})$}\revision{the proof} and $0.2\cdot T'$-th model performance in \deleted{$\Proof(\Adv, f_{W_T})$}\revision{the spoof}.
It shows that the model performance in \deleted{$\Proof(\Prov, f_{W_T})$}\revision{the proof} and \deleted{$\Proof(\Adv, f_{W_T})$}\revision{the spoof} are in similar trends.

\revision{
\Paragraph{Non-overlapping datasets.}
In~\cite{PoL}, \textlabel{it is assumed that}{r3:1:1} $\Adv$ has full access to the training dataset and can modify it.
An implicit assumption is that $\Verif$ does not know the dataset beforehand, otherwise the attack can be easily defended by checking the integrity of the dataset.
This assumption is realistic.
Consider the scenario where two hospitals share data with each other, so that they can train models (separately) for online diagnosis. 
Suppose hospital-A trains a good model, and hospital-B (which is the attacker) wants to claim the ownership of this model. 
Then, hospital-A provides a PoL proof and hospital-B provides a PoL spoof. 
In this example, $\Adv$ knows the training dataset but $\Verif$ does not.
}

\revision{
\textlabel{However}{r3:2:1}, $\Verif$ could still check if one dataset is an adversarial perturbation of the other.
This may introduce a new arms race: $\Adv$ would need to evade both PoL and this second verifier.
\textlabel{To this end}{r3:3:1}, we evaluate the effectiveness of our attacks on two non-overlapping datasets. 
Namely, we split the dataset CIFAR-10 into two non-overlapping datasets: $D_1$ and $D_2$.
$\Prov$ generates the PoL proof from $D_1$ and $\Adv$ generates the PoL spoof from $D_2$.
In this case, $\Verif$ can never know which dataset is an adversarial perturbation.
\fig~\ref{fig:Attack2_non_overlapping_eps} and \fig~\ref{fig:Attack3_non_overlapping_eps} show that we can still spoof PoL when the two training datasets are non-overlapping.
This is intuitively true; 
even for the same dataset, the $\mathtt{getBatch}$ function randomly samples data points, hence the data points used for spoofing are w.h.p. different from those used for generating PoL. 
That means our attack could always succeed with a different dataset. 
}

\begin{figure*}
    \centering
    \subfigure[Normalized reproduction error in $l_1$]{
    \includegraphics[width=0.31\linewidth]{ 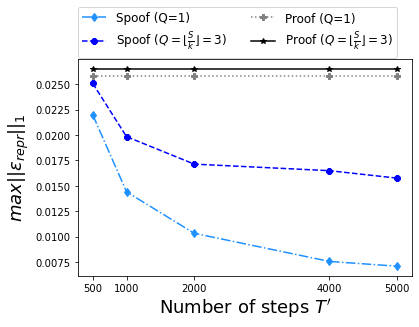}\label{fig:5.1}}
    \subfigure[Normalized reproduction error in $l_2$]{
    \includegraphics[width=0.31\linewidth]{ 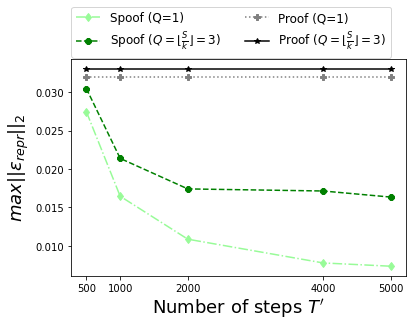}\label{fig:5.2}}
    \subfigure[Normalized reproduction error in $l_{\infty}$]{
    \includegraphics[width=0.31\linewidth]{ 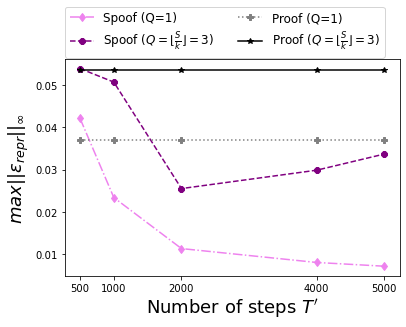}\label{fig:5.3}}
    \\
    \subfigure[Normalized reproduction error in $cos$]{
    \includegraphics[width=0.31\linewidth]{ 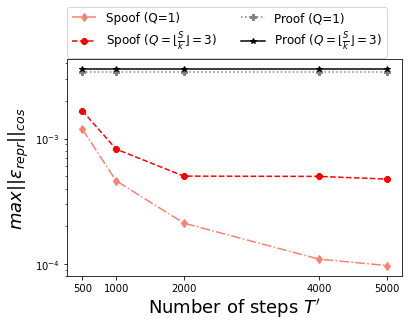}\label{fig:5.4}}
    \subfigure[Spoof generation time.]{
    \includegraphics[width=0.31\linewidth]{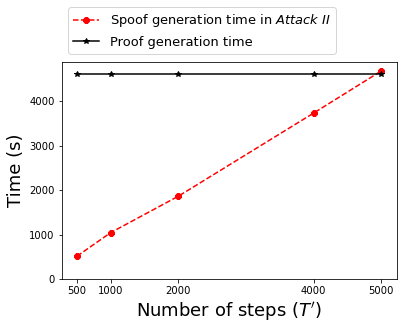}\label{fig:5.5}}
    \subfigure[Spoof size.]{
    \includegraphics[width=0.31\linewidth]{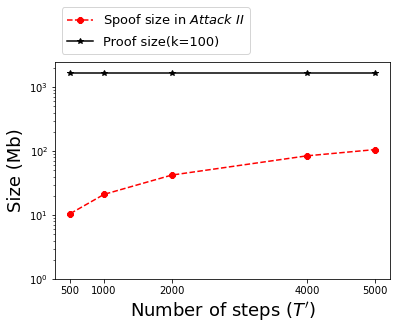}\label{fig:5.6}}        
  \caption{Attack II on non-overlapping CIFAR-10}
    \label{fig:Attack2_non_overlapping_eps}
\end{figure*}

\begin{figure*}
    \centering
    \subfigure[Normalized reproduction error in $l_1$.]{
    \includegraphics[width=0.31\linewidth]{ 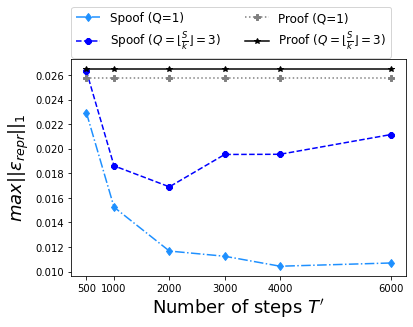}\label{fig:9.1}}
    \subfigure[Normalized reproduction error in $l_2$.]{
    \includegraphics[width=0.31\linewidth]{ 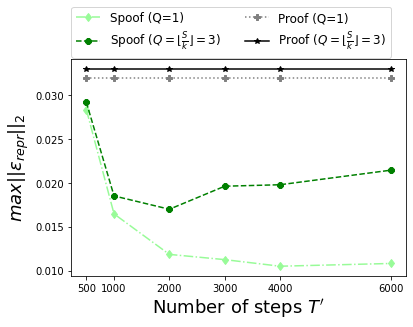}\label{fig:9.2}}
    \subfigure[Normalized reproduction error in $l_{\infty}$.]{
    \includegraphics[width=0.31\linewidth]{ 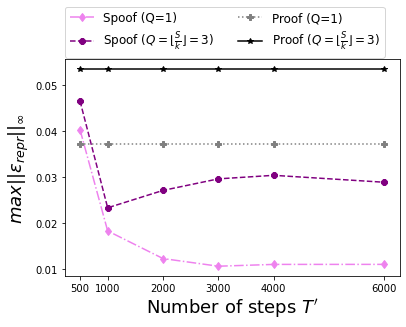}\label{fig:9.3}}
    \subfigure[Normalized reproduction error in $cos$.]{
    \includegraphics[width=0.31\linewidth]{ 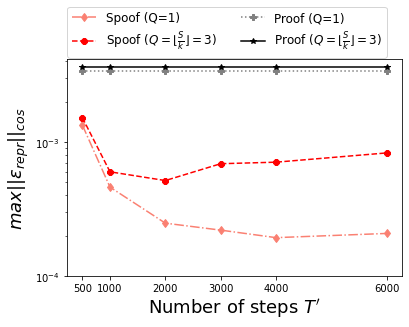}\label{fig:9.4}}
    \subfigure[Spoof generation time.]{
    \includegraphics[width=0.31\linewidth]{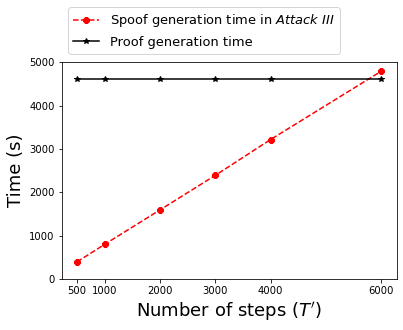}\label{fig:9.5}}
    \subfigure[Spoof Size.]{
    \includegraphics[width=0.31\linewidth]{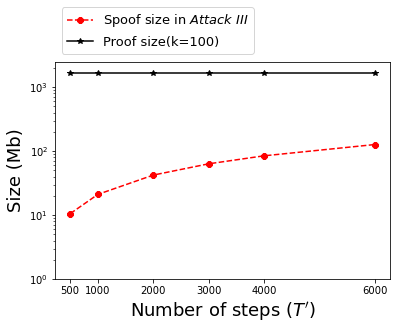}\label{fig:9.6}}
    
    \caption{Attack III on non-overlapping CIFAR-10.}
    \label{fig:Attack3_non_overlapping_eps}
\end{figure*}

\revision{
\Paragraph{Different hardware.}
\textlabel{Next}{e:2:1}, we show the spoof generation time is smaller than the proof generation time even if $\Adv$ uses a less powerful hardware then $\Prov$.
We pick four kinds of hardware and rank them in ascending order according to their computing power: 
{\sf Tesla T4}, {\sf Tesla P40}, {\sf GeForce RTX 2080Ti} and {\sf Tesla V100}. 
We run Attack~II and Attack~III on all these hardware.
In addition, we use the most powerful hardware (i.e., {\sf Tesla V100}) to generate the PoL proof.
\fig~\ref{fig:handwares} (in Appendix) shows that the attacks can succeed even if we use the least powerful hardware to generate spoof and use the most powerful hardware to generate proof.
}

\Paragraph{Summary.}
In summary, both Attack~II and Attack~III can successfully spoof a PoL when $T'$ is within a certain range.
\revision{\textlabel{The results}{r2:2:3} show that $T'$ can be large enough to be in the same level with $T$.
Consequently, $\Verif$ cannot simply detect the attack by setting a lower bound on the number of training steps.
}
\textlabel{ }{r2:2:2}
\deleted{
We further remark that our attacks are slow mostly because they are not well-engineered as normal training, e.g., the optimizers in Tensorflow and Pytorch adopt a bunch of tricks to accelerate training.
Therefore, the number of steps $T'$ could be potentially increased, making our attack more imperceptible.}

\fig~\ref{fig:ad_exp} (in Appendix) shows some randomly picked images before and after running our attacks.
The differences are invisible.


\section{Countermeasures}
\label{sec:discuss}

In this section, we provide two potential countermeasures for our attacks.

\deleted{
Model distance.}\textlabel{ }{r1:4:1}
\deleted{Recall that Attack~II and Attack~III require: $d(W_t, W_{t-k}) \leq \delta$.
However, in our experiments, we found that such distances in real training are usually larger. 
Therefore, $\Verif$ could set a bound for the distance between each $(W_t, W_{t-k})$.
Notice that this countermeasure is invalid for Attack I.}

\deleted{
Integrity of training dataset.
In PoL, it is assumed that $\Adv$ has full access to the training dataset $D$ and can modify it. 
If $\Adv$ has full access to $D$, most probably $D$ is a public dataset. 
Then, $\Verif$ can easily verify the integrity of $D$ so that $\Adv$ cannot modify it.
If $D$ is a private dataset, then $\Adv$ can neither access nor modify it.
That means we could change the threat model of PoL to make the attack more difficult.
However, $\Adv$ could still spoof $W_T$ using a totally different dataset $D'$ and claim that $D'$ is the real training dataset of $W_T$. }


\deleted{
Checkpointing interval and batch size.} \textlabel{ }{r2:3:5}
\deleted{
Recall that Attack~III for CIFAR-100 has to use either a small $k$ or a small batch size due to memory limitation.
Then, $\Verif$ could set bounds for both $k$ and the batch size.
However, this countermeasure is invalid for either Attack~I or Attack~II.
When $\Adv$ has a large GPU memory, it is invalid for Attack~III either.}


\textlabel{ }{r2:2:1}
\deleted{Number of training steps. Recall that Attack~II requires $T'<30$ to succeed and Attack~III requires $T'<300$. 
Then, $\Verif$ could set a bound for the number of training steps. 
However, as we mentioned, our attacks are slow mostly because they are not well-engineered as normal training.
We could potentially accelerate our attacks using the the state-of-the-art optimizers.
Then, bounding the number of of training steps will be no longer useful.
Furthermore, this countermeasure is invalid for Attack~I.}

\Paragraph{Selection of threshold.}
As we observed in our experiments, $\varepsilon_{\mathit{repr}}$ in the early training stage is usually larger than that in the later stage, because the model converges in the later stage. 
On the other hand, $\varepsilon_{\mathit{repr}}$ remains in the same level in our spoof.
Then, it is unreasonable for $\Verif$ to set a single verification threshold for the whole training stage.
A more sophisticated way would be dynamically choosing the verification thresholds according to the stage of model training:  choose larger thresholds for the early stage, and choose smaller thresholds for the later stage.
However, we can also set $d(W_t, W_{t-k})$ closer in the later stage to circumvent this countermeasure.

\revision{
\Paragraph{VC-based PoL.}
\textlabel{Verifiable computation}{r4:3:1} (VC) allows a delegator to outsource the execution of a complex function to some workers, which return the execution result together with a VC-proof;
the delegator can check the correctness of the returned results by verifying the VC-proof, which requires less work than executing the function.
There are many VC schemes such as SNARK~\cite{Pinocchio,libsnark}, STARK~\cite{libstark} etc; most of them require $O(n\log n)$ computational complexity to generate the proof, where $n$ is the number of gates in the function.
We can use VC to build a secure PoL mechanism:
during model training, $\Prov$ generates a VC-proof, proving that the final model $W_T$ was resulted by running the training algorithm on the initial model $W_0$ and the dataset $D$;
when the model ownership is under debate, $\Prov$ can show the VC-proof to $\Verif$.
To spoof, $\Adv$ has to devote $O(n\log n)$ computation to generate a valid VC-proof,  which is almost equal as $\Prov$.
This mechanism is valid, but it will introduce an overwhelming overhead.
}

\section{Related Work}
\label{sec:related}

\Paragraph{Adversarial examples.}
When first discovered in 2013~\cite{intri2013}, adversarial examples are images designed intentionally to cause deep neural networks to make false predictions. 
Such adversarial examples look almost the same as  original images, thus they show vulnerabilities of deep neural networks~\cite{goodfellow2014explaining}. 
Since then, it becomes a popular research topic and has been explored extensively in both attacks~\cite{dong2018boosting, su2019one} and defences~\cite{papernot2016distillation,buckman2018thermometer,Bhagoji2018Enhancing,zheng2016improving,wang2017learning,luo2016foveationbased}. 


Adversarial examples are generated by solving:
$$\mathbf{R}=\mathop{\arg\min}_{\mathbf{R}}L(f_{W}(\mathbf{X}+\mathbf{R}),\mathbf{y}')+\alpha||\mathbf{R}||,$$
where $\mathbf{y}'$ is a label that is different from the real label $\mathbf{y}$ for $\mathbf{X}$.
Then, the noise $\mathbf{R}$ can fool the model to predict a wrong label (by minimizing the loss function) and pose little influence on the original instance $\mathbf{X}$.

Recall that the objective of Attack~II and Attack~III is to minimize the gradients computed by ``adversarial examples''. 
Therefore, it can be formulated as: $$\mathbf{R}=\mathop{\arg\min}_{\mathbf{R}}L(f_{W}(\mathbf{X}+\mathbf{R}),\mathbf{y})+\alpha||\mathbf{R}||.$$
This is identical to the objective of finding an adversarial example.
An adversarial example aims to fool the model whereas our attacks aim to update a model to itself.
Nevertheless, they end up at the same point, which explains the effectiveness of Attack II and Attack III.




\Paragraph{Unadversarial examples.}
\textlabel{Unadversarial examples}{r1:1:1}~\cite{salman2021unadversarial} target the scenarios where a system designer not only trains the model for predictions, but also controls the inputs to be fed into that model.
For example, a drone operator who trains a landing pad detector can also modify the surface of the landing pads.
Unadversarial examples are generated by solving:
$$\mathbf{R}=\mathop{\arg\min}_{\mathbf{R}}L(f_{W}(\mathbf{X}+\mathbf{R}),\mathbf{y}),$$
so that the new input $(\mathbf{X}+\mathbf{R})$ can be recognized better by the model $f_{W}$.

This is similar to the objective function of our attacks besides the regularization term we use to minimize the perturbations. 
The authors in~\cite{salman2021unadversarial}  demonstrate the effectiveness unadversarial examples via plenty of experimental results, which can also explain the success of our attacks.

\Paragraph{Deep Leakage from Gradient.}
In federated learning~\cite{mcmahan17a, bonawitz2019towards, li2020federated, kairouz2019advances}), 
it was widely believed that shared gradients will not leak information about the training data.
However, Zhu et al.~\cite{zhu2020deepleakage} proposed ``Deep Leakage from Gradients'' (DLG), where the training data can be recovered through gradients matching.
Specifically, after receiving gradients from another worker, the adversary feeds a pair of randomly initialized dummy instance $(X, y)$ into the model, and obtains the dummy gradients via back-propagation.
Then, they update the dummy instance with an objective of minimizing the distance between the dummy gradients and the received gradients:
\begin{equation*}
    \begin{aligned}
  \XX, y &= \mathop{\arg\min}_{\XX, y}||\triangledown W'-\triangledown W||^2\\&= \mathop{\arg\min}_{\XX, y}||\frac{\partial L(f_W(\XX), y)}{\partial W}-\triangledown W||^2
    \end{aligned}
\end{equation*}
After a certain number of steps, the dummy instance can be recovered to the original training data. 

Our attacks were largely inspired by DLG: 
an instance can be updated so that its output gradients can match the given gradients.
The difference between DLG and our work is that DLG aims to recover the training data from the gradients, whereas we want to create a perturbation on a real instance to generate specific gradients.

\section{Conclusion}
\label{sec:conc}

In this paper, we show that a recently proposed PoL mechanism is vulnerable to ``adversarial examples''.
Namely, in a similar way as generating adversarial examples, we could generate a PoL spoof  with significantly less cost than generating a proof by the prover.
We validated our attacks by conducting  experiments extensively. 
In future work, we will explore more effective attacks.

\section*{Acknowledgments}

The work was supported in part by 
National Key Research and Development Program of China under Grant 2020AAA0107705,
National Natural Science Foundation of China (Grant No. 62002319, 11771393, U20A20222) and
Zhejiang Key R\&D Plans (Grant No. 2021C01116).

\balance
\bibliographystyle{plain}
\bibliography{references}

\appendix

\section{Appendix}
\clearpage
\begin{figure*}[ht]
    \centering
    \subfigure[Intermediate models accuracy on CIFAR-10]{
    \includegraphics[width=0.3\linewidth]{ 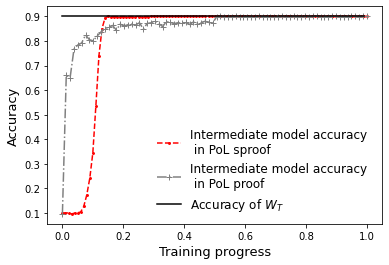}}
    \subfigure[Intermediate models accuracy on CIFAR-100]{
    \includegraphics[width=0.3\linewidth]{ 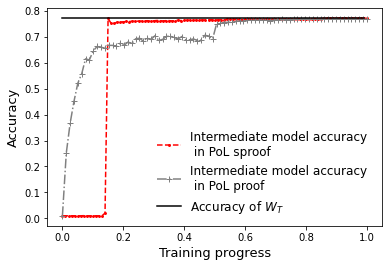}}
    \subfigure[Intermediate models accuracy on ImageNet]{
    \includegraphics[width=0.3\linewidth]{ 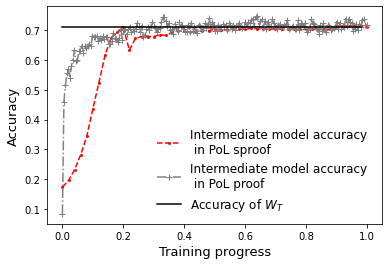}}
    
    \caption{Intermediate model accuracy for 
    Attack~III. 
    \revision{(The $x$-axis presents the progress of training. 
    For example when $x=0.2$, the corresponding $y$ represents the $0.2\cdot T$-th model performance in PoL proof and $0.2\cdot T'$-th model performance in PoL spoof)}}
    \label{fig:intermediate_model_accuracyl}
\end{figure*}

\begin{figure*}
    \centering
   \subfigure[CIFAR-10.]{
    \includegraphics[width=0.32\linewidth]{ 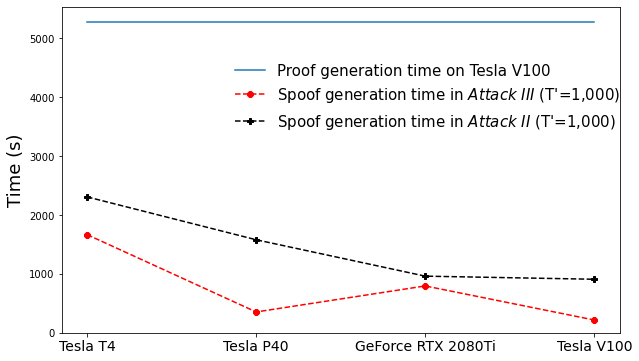}}
    \subfigure[CIFAR-100.]{
    \includegraphics[width=0.32\linewidth]{ 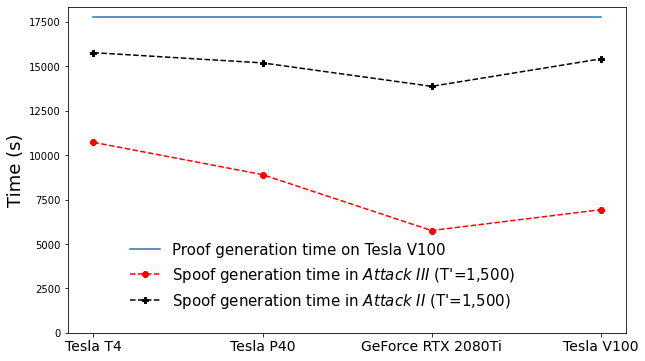}}
    \subfigure[ImageNet]{
    \includegraphics[width=0.32\linewidth]{ 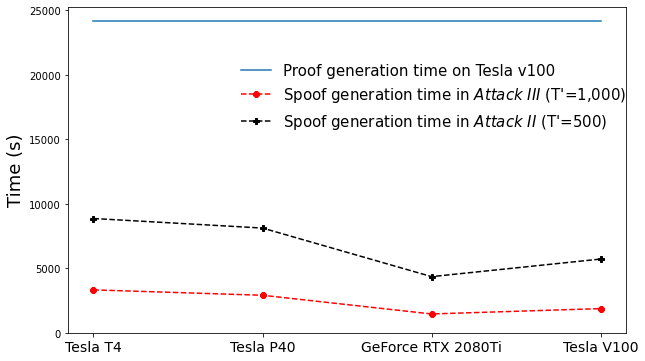}}
    \caption{Spoof generation time with different GPUs. 
    {\sf Tesla T4} is the least powerful one and  {\sf Tesla V100} is the most powerful one. We use  {\sf Tesla V100} to generate the proof.
    }
    \label{fig:handwares}
\end{figure*}

\begin{figure*}
    \centering
    \subfigure[CIFAR-10]{
    \includegraphics[width=0.3\linewidth]{ 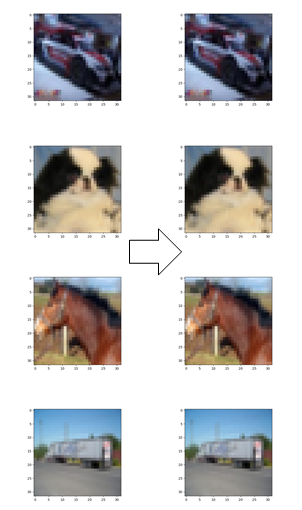}}
    \subfigure[CIFAR-100]{
    \includegraphics[width=0.3\linewidth]{ 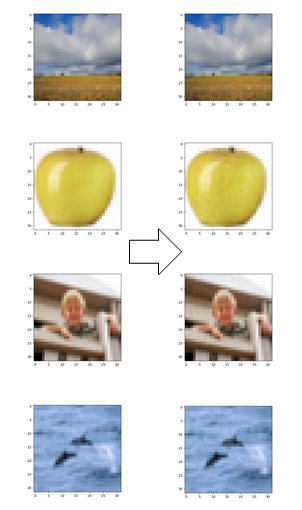}}
    \subfigure[ImageNet]{
    \includegraphics[width=0.3\linewidth]{ 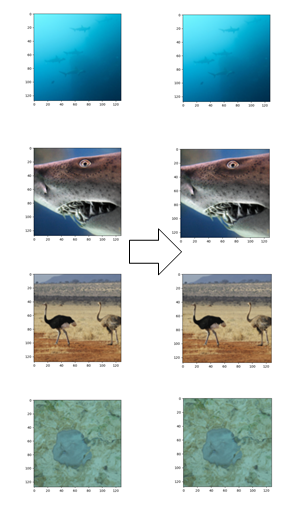}}
    \caption{``Adversarial examples'' generated by Attack~III. The original images are on the left-hand side and the noised images are on the right-hand side.}
    \label{fig:ad_exp}
\end{figure*}

\end{document}